\documentclass{article}
\usepackage[T1]{fontenc}
\usepackage[latin1]{inputenc}
\usepackage{graphicx}
\usepackage{amsmath}
\usepackage{amsthm}
\usepackage{amsfonts}
\usepackage{amssymb}
\usepackage{epsfig}
\usepackage{enumerate}
\usepackage{bm,bbm}
\usepackage{fullpage}
\usepackage{float}
\newtheorem{theorem}{\textbf{Theorem}}[section]
\newtheorem{lemma}{\textbf{Lemma}}[section]

\newtheorem{definition}{\textbf{Definition}}[section]
\newtheorem{corollary}{\textbf{Corollary}}[section]
\newtheorem{remark}{\textbf{Remark}}[section]
\newtheorem{assumption}{\textbf{Assumption}}[section]

\newcommand{\cF}{\mathcal{F}}
\newcommand{\cC}{\mathcal{C}}

\newcommand{\R}{\mathbb{R}}

\newcommand{\E}{\mathbb{E}}

\newcommand{\vi}{\mathbbm{x}}
\newcommand{\sgn}{\mathrm{sgn}}

\newcommand{\strat}{({\bm \xi},\mathcal{T})}
\newcommand{\stratd}{({\bm x}, {\bm t})}
\newenvironment{proofT}[1]{\textbf{Proof of Theorem $\mathbf{\ref{#1}}$:}}{\hfill$\Box$
\\\bigskip}

\newenvironment{proofC}[1]{\textbf{Proof of Corollary $\mathbf{\ref{#1}}$:}}{\hfill$\Box$
\\\bigskip}

\title{Optimal execution and price manipulations in time-varying limit order books}
\author{Aurélien Alfonsi and José Infante Acevedo\\
  Université Paris-Est, CERMICS,\\ Project team MathRISK ENPC-INRIA-UMLV\\  Ecole des Ponts,   6-8 avenue
Blaise Pascal,  77455
Marne-la-vall\'{e}e, France\\  {\tt alfonsi@cermics.enpc.fr, \ jose-infante.acevedo@cermics.enpc.fr
}}

\begin{document}  
\maketitle
{\noindent {\bf Abstract:} This paper focuses on an extension of the Limit
  Order Book (LOB) model with general shape introduced by Alfonsi, Fruth and Schied~\cite{AFS}. Here, the additional feature allows a time-varying LOB
  depth. We solve the optimal execution problem in this framework for both
  discrete and continuous time strategies. This gives in particular sufficient
  conditions to exclude Price Manipulations in the sense of Huberman and
  Stanzl~\cite{HS} or Transaction-Triggered Price Manipulations (see Alfonsi,
  Schied and Slynko~\cite{ASS}). These conditions give interesting qualitative
insights on how market makers may create or not price manipulations. }\\


{\it AMS Class 2010:} 91G99, 91B24, 91B26, 49K99 \\
 
{\noindent {\bf Acknowledgements:} This work has benefited from the support of the Eurostars E!5144-TFM
project and of the ``Chaire Risques Financiers'' of Fondation du Risque. José
Infante Acevedo is grateful to AXA Resarch~Fund for his doctoral fellowship.}
%\newpage
%\tableofcontents
%\newpage

\section*{Introduction}

It is a rather standard assumption in finance to consider an infinite
liquidity. By infinite liquidity, we mean here that the asset price  is given by a a single value,
and that one can buy or sell any quantity at this price without changing the
asset price. This assumption is in particular made in the Black and Scholes
model~\cite{BS}, and is often made as far as derivative pricing is
concerned. When considering portfolio over a large time horizon, this
approximation is relevant since one may split orders in small ones along the
time and reduces one's own impact on the price. At most, the lack of liquidity
can be seen as an additional transaction cost. This issue has been broadly
investigated in the literature, see Cetin, Jarrow and Protter~\cite{CJP} and
references within.

If we consider instead brokers that have to trade huge volumes over a short
time period (some hours or some days), we can no longer neglect the price
impact of trading strategies. We have to focus on the market microstructure
and model how prices are modified when buy and sell orders are
executed. Generally speaking, the quotation of an asset is made through a
Limit Order Book (LOB) that lists all the waiting buy and sell orders on this
asset. The order prices have to be a multiple of the tick size, and orders at
the same price are arranged in a First-In-First-Out stack. 
The bid (resp. ask) price is the price of the highest waiting buy
(resp. lowest selling buy) order. Then, it is possible to buy or sell the asset in
two different ways: one can either put a limit order and wait that this order
matches another one or put a market order that consumes the cheapest limit
orders in the book. In the first way, the transaction cost is known but the
execution time is uncertain. In the second way, the execution is immediate
(provided that the book contains enough orders). The price per share instead
depends on the order size. For a buy (resp. sell) order,
the first share will be traded at the ask (resp. bid) price while the last one
will be traded some ticks upper (resp. lower) in order to fill the order
size. The ask (resp. bid) price is then modified accordingly. 

%% Let us consider for example a buy order whose size
%% is in~$(N_{k-1},N_k]$, where $N_k$ is the number of limit sell orders with
%% price smaller than the ask price plus $k$~kicks ($N_{-1}=0$ by convention). Then,
%% the first share will be bought at the ask price while the last share will be bought at
%% the ask price plus $k$~ticks.  

The typical issue on a short time scale is the optimal execution problem:
on given a time horizon, how to buy or sell optimally a given amount of
assets? As pointed in Gatheral~\cite{Gatheral} and Alfonsi, Schied and
Slynko~\cite{ASS}, this problem is closely related to the market viability and
to the existence of price manipulations.
Modelling the full LOB dynamics is not a trivial issue, especially if one
wants to keep tractability to solve then the optimal execution
problem. Instead, simpler models called market impact models have been
proposed. These models only describe the dynamics of one asset price and model
how the asset price is modified by a trading strategy. Thus, Bertsimas and
Lo~\cite{BL},  Almgren and Chriss~\cite{AC}, Obizhaeva and Wang~\cite{OW} have proposed
different models where the price impact is proportional to the trading
size, in which they solve the optimal execution problem. However, some
empirical evidence on the markets show that the price 
impact of a trade is not proportional to its size, but is rather proportional
to a power of its size (see for example Potters and Bouchaud~\cite{PB}, and
references within). With this motivation in mind, Gatheral~\cite{Gatheral} has
suggested a nonlinear price impact model. In the same direction, Alfonsi,
Fruth and Schied~\cite{AFS} have derived a  price impact model from a simple LOB
modelling. Basically, the LOB is modelled by a shape function that describes
the density of limit orders at a given price. This model has then been studied
further by Alfonsi and Schied~\cite{AS} and Predoiu, Shaikhet and
Shreve~\cite{PSS}.

The present paper extends this model by letting the LOB shape function vary 
along the time. Beyond solving the optimal execution problem in a more general
context, our goal is to understand how the dynamics of the LOB may create or
not price manipulations. Indeed, a striking result in~\cite{AFS,AS} is that
the optimal execution strategy is made with trades of same sign, which excludes
any price manipulation. This result holds under rather general assumptions on
the LOB shape function, when the LOB shape does not change along the
time. Instead, we will see in this paper that a time-varying LOB may induce price
manipulations and we will derive sufficient conditions to exclude them. These
conditions are not only interesting from a theoretical point of view. They
give a qualitative understanding on how price manipulations may occur when
posting or cancelling limit orders. While preparing this work, Fruth,
Schöneborn and Urusov~\cite{FSU} have presented a
paper where this issue is addressed for a block-shaped LOB, which amounts to a
proportional price impact. Here, we get back their result and extend them to
general LOB shapes and thus nonlinear price impact. The other contribution of this paper is that we solve
 the optimal execution in a continuous time setting
while~\cite{AFS,AS} mainly focus on discrete time strategies. This is in
particular much more suitable to state the conditions that exclude price
manipulations.

\section{Market model and the optimal execution problem}

\subsection{The model description}

The problem that we study in this paper is the classical optimal
execution problem. To deal with this problem, we consider in this paper
a framework which is a natural extension of the model
proposed in Alfonsi, Fruth and Schied~\cite{AFS} and developed by Alfonsi and
Schied~\cite{AS} and Predoiu, Shaikhet and Shreve~\cite{PSS}. The additional feature that we
introduce here is to allow a time varying depth of the
order book.  We consider a large trader who wants to
liquidate a portfolio of $\vi$ shares in a time period of $[0,T]$.
In order liquidate these $\vi$ shares, the large trader uses only market orders, that is
buy or sell orders that are immediately executed at the best available current
price. Thus, our large trader cannot put limit orders. A long
position $\vi>0$ will correspond to a sell program while a short position
$\vi<0$ will stand for a a buy strategy. The optimal execution problem
consists in finding the optimal trading strategy that minimizes the expected
cost of the large trader.

We assume that the
price process without the large trader would be given by a rightcontinuous
martingale $(S^0_t,t\ge 0)$ on a given filtered probability space
$\left(\Omega,(\mathcal{F}_t),\mathcal{F},\mathbb{P}\right)$. The actual price
process~$(S_t,t\ge 0)$ that takes into account the trades of the large trader
is defined by:
\begin{equation}\label{actualprice}
  S_t=S^0_t+D_t, \ t\ge 0.
\end{equation}
Thus, the process $(D_t,t\ge 0)$ describes the price impact of the large
trader. We also introduce the process $(E_t,t\ge 0)$ that describes the volume
impact of the large trader. If the large trader 
puts a market order of size $\xi_t$ ($\xi_t>0$ is a buy order and $\xi_t<0$ a sell order), the volume impact process changes from $E_t$ to:
\begin{equation}
  \label{eq:E_t_plus}
  E_{t+}:=E_t+\xi_t.
\end{equation}
When the large trader is not active, its volume impact $E_t$ goes back to~$0$. We
assume that it decays exponentially with a deterministic time-dependent rate
$\rho_t > 0$ called resilience, so that we have:
\begin{equation}
  \label{eq:E_dynamic}
  dE_t=-\rho_tE_tdt.
\end{equation}

We now have to specify how the processes $D$ and $E$ are
related. To do so, we suppose a continuous
distribution buy and sell limit orders around the unaffected price
$S^0_t$: for $x\in \R$, we assume that the number of limit orders available
between prices $S^0_t+x$ and $S^0_t+x+dx$ is given by
$\lambda(t)f(x)dx$. These orders are sell orders if $x\ge D_t$ and buy orders
otherwise. The
functions  $f:\mathbb{R}\mapsto(0,\infty)$ and
$\lambda:[0,T]\mapsto(0,\infty)$ are assumed to be continuous, and represent
respectively the LOB shape and the depth of the order book.  We define the antiderivative of the function $f$,
$F(y):=\int^y_0f(x)dx,\;\; y\in\mathbb{R},$
and assume that
\begin{equation}\label{infinite_d_ordres}
\underset{x\rightarrow -\infty}{\lim }F(x) = -\infty  \text{ and }\underset{x\rightarrow \infty}{\lim }F(x) = \infty,
\end{equation}
which means that the book contains an infinite number of limit buy and sell orders. 
Thus, we set the following relationship between the volume impact $E_t$ and the price impact $D_t$:
\begin{equation*}
  \int^{D_t}_{0}\lambda(t)f(x)dx=E_t,
\end{equation*}
or equivalently,
\begin{equation}
  \label{eq:relation_E_D}
  E_t=\lambda(t)F(D_t)\text{ and }D_t=F^{-1}\left(\frac{E_t}{\lambda(t)}\right).
\end{equation}
Within this framework, a large trade~$\xi_t$ changes $D_t$ to~$D_{t+}=F^{-1}\left(\frac{E_t+\xi_t}{\lambda(t)}\right)$ and has the cost
\begin{equation}\label{def_costpi}
  \int_{D_t}^{D_{t+}}(S^0_t+x)\lambda(t)f(x)dx=\xi_tS^0_t+\int_{D_t}^{D_{t+}}
  \lambda(t)x f(x)dx:=\pi_t(\xi_t).
\end{equation}
Throughout the paper, we assume that $\lambda$ is $\mathcal{C}^2$ and set
$\eta_t=\frac{\lambda'(t)}{\lambda(t)}$. Thus, we have
$$\lambda(t)=\lambda(0)\exp\left(\int_0^t \eta_u du\right),$$ and $t\mapsto
\eta_t$ is  $\mathcal{C}^1$. Similarly, we assume that $t\mapsto
\rho_t$ is  $\mathcal{C}^1$.

Now, let us observe that we have assumed that the volume impact decays
exponentially when the large trader is inactive. Other choices are of course
possible, and a natural one would be to assume that the price impact decays
exponentially
\begin{equation}
  \label{eq:D_dynamic}
  dD_t=-\rho_tD_tdt,
\end{equation}
which amounts to assume that
$dE_t=\eta_tE_tdt-\rho_t\lambda(t)f(F^{-1}(E_t/\lambda(t)))F^{-1}(E_t/\lambda(t))dt$ by~\eqref{eq:relation_E_D}.

\begin{definition}
The dynamics of ``model~$V$'' with volume impact reversion is the one given
by~\eqref{actualprice}, \eqref{eq:E_t_plus}, \eqref{eq:E_dynamic} and
\eqref{eq:relation_E_D}. The dynamics of ``model~$P$'' with price impact reversion is the one given
by~\eqref{actualprice}, \eqref{eq:E_t_plus}, \eqref{eq:D_dynamic} and
\eqref{eq:relation_E_D}. In both models, we assume that the market is at
equilibrium at time~$0$, i.e. $E_0=D_0=0$.
\end{definition}

\begin{remark}\label{Rk_marketmakers}
 Though being simplistic, this model describes
through $\rho_t$ and $\lambda(t)$ the two different ways that market
makers have to put (or cancel) limit orders: it is either possible to pile orders at an existing price or
to put orders at a better price than the existing ones. Thus, $\lambda(t)$
describes how market makers pile orders while $\rho_t$ describes the rate at
which new orders appear at a better price. 
Basically, one may think these
functions one-day periodic, with relative high values at the opening and the
closing of the market and low values around noon. The particular case $\lambda \equiv 1$ corresponds to the model introduced
by Alfonsi, Fruth and Schied~\cite{AFS} for which new orders can only
appear at a better price.
\end{remark}

%% Examples : Block-shape, power law...
%%  Remark: two-sided LOB

\subsection{The optimal execution problem, and price manipulation strategies}\label{ssec_oep}

We focus on the optimal liquidation of a portfolio with $\vi$
shares by a large trader who can place market orders over a period of time
$[0,T]$. Thus, $\vi>0$ (resp. $\vi<0$) corresponds to a selling (resp. buying)
strategy.

We first consider discrete strategies and assume that at most $N+1$ trades can occur. An admissible strategy
will be then described by an increasing sequence $\tau_0=0\le \dots \le
\tau_N=T$ of stopping times and random variables $\xi_0,\dots,\xi_N$ ($\xi_i$
stands for the trading size at time~$\tau_i$) such that
\begin{itemize}
\item $\vi+\sum_{i=0}^N \xi_i= 0$, i.e. the trader liquidates indeed $\vi$
  shares,
\item $\xi_i$ is $\mathcal{F}_{\tau_i}$-measurable,
\item $\exists M \in \R , \forall 0\le i\le N, \xi_i\ge M$, a.s.
\end{itemize}
The expected cost of an admissible strategy $({\bm \xi},\mathcal{T})$ with $
{\bm \xi}=(\xi_0,\dots,\xi_N)$ and $\mathcal{T}=(\tau_0, \dots ,
\tau_N)$ is given by
\begin{equation}\label{cost_discrete}\mathcal{C}({\bm \xi},\mathcal{T})=\E \left[ \sum_{i=0}^N \pi_{\tau_i}(\xi_i)
\right],
\end{equation}
where $\pi_{\tau_i}(\xi_i)$ stands for the cost of the $i$-th trade, and is
defined by~\eqref{def_costpi} in models~$V$ or~$P$.  The goal of the large trader
is then to minimize this expected cost among the admissible strategies.

We also consider continuous time trading strategy and make the same
assumptions as Gatheral et al.~\cite{GSS}. An admissible strategy $(X_t)_{t \ge
  0}$ is a stochastic process such that
\begin{itemize}
\item $X_{0}=\vi$ and $X_{T+}= 0$, 
\item $X$ is $(\cF_t)$-adapted and leftcontinuous,
\item the function $t\in[0,T+]\mapsto X_t$ has finite and a.s. bounded total variation.
\end{itemize}
The process $X_t$ describes the number of shares that remains to liquidate at
time~$t$. Thus, the discrete time strategy above corresponds to $X_t=\vi+\sum_{i=0}^N \xi_i \mathbf{1}_{\tau_i < t}$, and the three assumptions
on $({\bm \xi},\mathcal{T})$ precisely give the ones on~$X$. Let us observe
that processes~$E$ and~$D$ are also leftcontinuous since we have in model~$V$
(resp. model~$P$):
\begin{equation}\label{dynamic_modelsVP}
  dE_t=dX_t-\rho_tE_tdt,\ (\text{resp.}\
dE_t=dX_t+\eta_tE_tdt-\rho_t\lambda(t)f(F^{-1}(E_t/\lambda(t)))F^{-1}(E_t/\lambda(t))dt
).
\end{equation}
We want now to write the cost associated to the strategy~$X$. To do so,
we introduce the following notations
\begin{equation}\label{G_def}  
x\in \R, \tilde{F}(x)=\int^{x}_{0}yf(y)dy, \ G(x)=\tilde{F}\left(F^{-1}(x)\right),
\end{equation}
so that $\pi_t(dX_t)=S^0_tdX_t+\lambda(t)[G\left(\frac{E_t+dX_t}{ \lambda(t)}
    \right)-G\left(\frac{E_t}{ \lambda(t)}    \right)]$. Since $G'=F^{-1}$,
the cost of an admissible strategy is given by:
\begin{equation}\label{cost_continuous}
  \mathcal{C}(X)=\E \left[ \int_0^T \left[S^0_t+ 
  F^{-1}\left(\frac{E_t}{ \lambda(t)} \right) \right]dX_t +\sum_{t\le T}
  \lambda(t)\left[G\left(\frac{E_t+\Delta X_t}{ \lambda(t)}
    \right)-G\left(\frac{E_t}{ \lambda(t)} \right) - F^{-1}\left(\frac{E_t}{
        \lambda(t)} \right) \Delta X_t \right]
\right],
\end{equation}
which coincides with~\eqref{cost_discrete} for discrete strategies. Here,
$\Delta X_t=X_{t+}-X_t$ denotes the jump of~$X$ at time~$t$ (jumps are
countable), and $dX_t$ stands for the signed measure on~$[0,T]$ associated to
$(X_t,0\le t\le T+)$ (a jump $\Delta X_T$ induces a Dirac mass in~$T$). If we
introduce the continuous part of~$X$, $X^c_t:=X_t-\sum_{0 \le s <t} \Delta
X_s$, we can rewrite the cost as follows:
\begin{equation}\label{cost_continuous2}
  \mathcal{C}(X)=\E \left[ \int_0^T \left[S^0_t+ 
  F^{-1}\left(\frac{E_t}{ \lambda(t)} \right) \right]dX^c_t +\sum_{t\le T}
 S^0_t\Delta X_t+  \lambda(t)\left[G\left(\frac{E_t+\Delta X_t}{ \lambda(t)}
    \right)-G\left(\frac{E_t}{ \lambda(t)} \right)  \right]
\right].
\end{equation}

The optimal execution problem is in fact closely related to questions around
market viability and arbitrage. We recall the definition of price manipulation
strategies introduced by Huberman and Stanzl~\cite{HS}.
\begin{definition}
A round trip is an admissible strategy~$X$ for $\vi=0$. A Price
Manipulation Strategy (PMS) in the sense of Huberman and Stanzl is a round
trip whose expected cost is negative, i.e. $\mathcal{C}(X) <0$.
\end{definition}
Heuristically, if a PMS exists, it could be repeated indefinitely and would lead to a
classical arbitrage (i.e. an almost sure profit) by a law of large
numbers. However, it has been pointed in Alfonsi et al.~\cite{ASS} the absence of PMS does
not ensure the market stability. In fact, in some PMS free models, the optimal
strategy to sell~$\vi$ shares consists in buying and selling successively a
much higher amount of shares. To correct this, they introduce the following definition.
\begin{definition}
A model admits transaction-triggered price manipulations (TTPM) if the
expected cost of a sell (buy) program can be decreased by intermediate buy
(sell) trades, i.e.
\begin{equation*}
\exists X\text{ admissible}, \ \mathcal{C}(X)  < \inf\left\{ \mathcal{C} (\tilde{X}) ,
  \tilde{X} \text{ is admissible and nonincreasing or nondecreasing} \right\}.
\end{equation*}
\end{definition}
It is rather natural choice to exclude TTPM: in presence of TTPM a large
trader would increase the traded volume to minimize its cost, which produce
noise and may yield to instability. Besides, the absence of TTPM implies 
the absence of PMS. The optimal strategy for buying $\varepsilon>0$ shares is
made only with intermediate buy trades and has thus a positive cost. Thus, by
some cost continuity that usually holds (this is the case for models~$V$
and~$P$), any round trip has a nonnegative cost. 

\begin{remark}\label{2sides} It is possible to define a two-sided limit order
  book model like in Alfonsi, Fruth and Schied~\cite{AFS} or Alfonsi and
  Schied~(\cite{AS}, Section 2.6). In such a model, bid and ask prices evolve
  as follows. A buy (resp. sell) order of the large trader shifts the ask
  (resp. bid) price and leaves the bid (resp. ask) price unchanged. When the large
  trader is idle, the shifts on the ask and bid prices goes back exponentially
  to zero, like in models~$V$ or $P$. As in~\cite{AFS,AS}, the two-sided
  model coincides with the model presented here when the large trader  puts
  only buy orders or only sell orders. In particular, the optimal strategies
  are the same in both models in absence of TTPM. 
\end{remark}

\section{Main results}

The first focus of this paper is to extend the results obtained in Alfonsi et al.~\cite{AFS,AS}
 and obtain the optimal execution strategies for LOB with a time
varying depth~$\lambda$. Doing so, our goal is also to better understand how this time
varying depth may create manipulation strategies. In fact, it was shown in~\cite{AFS}
and~\cite{AS} for $\lambda\equiv 1$ that under some general assumptions on the shape function~$f$,
there is an optimal liquidation strategy which is made only with sell
(resp. buy) orders
when $\vi>0$ (resp. $\vi<0$). Thus, there is no PMS nor TTPM when the LOB shape is
constant. This is a striking result, and one may wonder how this is modified
by changing slightly the assumptions. In Alfonsi, Schied and Slynko~\cite{ASS} is studied the case of a
block-shaped LOB, where the resilience is not exponential so that the market
has some memory of the past trades. Conditions on the market resilience are
given to exclude PMS and TTPM. Analogously, we want to obtain here conditions
on~$\lambda$ and $\rho$ that rules out such strategies. This is not only
interesting from a theoretical point of view. This will give also some noticeable
qualitative insights for market makers. In fact, for a market maker who places
and cancels significant limit orders, these conditions will indicate if he may
or not create manipulation strategies.

Before showing the results, it is worth to make further derivations on the
expected cost. Let us start with discrete strategies. By using the
martingale property on~$S^0$ and the assumptions on~${\bm \xi}$ made in
Section~\ref{ssec_oep}, we can show easily like in~\cite{AS} that
$$\mathcal{C}({\bm \xi},\mathcal{T})=-S^0_0\vi + \E \left[ \sum_{i=0}^N \int_{D_{\tau_i}}^{D_{\tau_i+}}
  \lambda(\tau_i)x f(x)dx \right]. $$
Then, it is easy to check that $\sum_{i=0}^N \int_{D_{\tau_i}}^{D_{\tau_i+}} \lambda(\tau_i)x f(x)dx  $
is a deterministic function of $({\bm \xi},\mathcal{T})$ in both volume impact
reversion and price impact reversion models. We respectively denote by
$C^V({\bm \xi},\mathcal{T})$ and $C^P({\bm \xi},\mathcal{T})$ this function
and get:
\begin{equation}\label{CM_def}\mathcal{C} \strat =-S^0_0\vi + \E \left[ C^M \strat \right],
\end{equation}
where $M \in \{V,P\}$ indicates the model chosen. Thus, if the function
$\stratd \mapsto C^M \stratd$ has a unique minimizer on $\{\stratd \in
\R^N\times \R^{N+1}, \sum_{i=1}^N x_i=-\vi, 0=t_0\le \dots\le t_N=T \}$, the
optimal strategy is deterministic and given by this minimizer. When $\lambda$
is constant, it is shown in~\cite{AS} that under some assumptions on~$f$
depending on the model chosen, the optimal time grid ${\bm t}^\star$ is homogeneous with
respect to~$\rho$, i.e. $\int_{t_i^\star}^{t_{i+1}^\star} \rho_sds
=\frac{1}{N}\int_{0}^{T} \rho_sds$. Instead, there is no such a simple
characterization for general $\lambda$, even in the block-shaped case. Thus, we
will focus on optimizing the trading strategy~${\bm \xi }$ on a fixed time
grid~${\bm t}$:
\begin{equation}\label{det_time_grid}
  {\bm t}=(t_0,\dots,t_N), \text{ such that $0=t_0< \dots < t_N=T$}.
\end{equation}
Last, we introduce the following notations that will be used throughout the paper:
\begin{equation}
\label{notations_a}
a_{i}=e^{-\int^{t_{i}}_{t_{i-1}}\rho_{u}du},\
\tilde{a}_i=\frac{a_i\lambda(t_{i-1})}{\lambda(t_i)}=e^{-\int^{t_{i}}_{t_{i-1}}(\rho_{u}+\eta_u)du}, \ 
\hat{a}_i=a_i\frac{\lambda(t_i)}{\lambda(t_{i-1})}=e^{-\int^{t_{i}}_{t_{i-1}}(\rho_{u}-\eta_u)du}, \  \ \ 1 \le i\le N. 
\end{equation}

Similarly in the continuous case, we get by using the martingale assumption
(see Lemma~2.3 in Gatheral, Schied and Slynko~\cite{GSS}) that $\E[\int_0^TS^0_tdX_t]=- \vi
S^0_0$. From~\eqref{cost_continuous} and~\eqref{cost_continuous2}, we get $\cC(X)=- \vi
S^0_0+\E[C^M(X)]$, where
$$C^M(X)=\int_0^T  
  F^{-1}\left(\frac{E_t}{ \lambda(t)} \right) dX^c_t +\sum_{t\le T}
  \lambda(t)\left[G\left(\frac{E_t+\Delta X_t}{ \lambda(t)}
    \right)-G\left(\frac{E_t}{ \lambda(t)} \right)\right].$$
Once again, $C^M$ is a deterministic function of the strategy~$X$ in both
models~$M \in \{V,P\}$, and it is sufficient to focus on its minimization.

\subsection{The block-shaped LOB case ($f\equiv 1$)}

In this section, we consider a shape function of the limit order book that has
the form $\lambda(t)$. This time-dependent framework generalizes the
block-shaped limit order book case studied by Obizhaeva and
Wang~\cite{OW} that consists in considering a uniform
distribution of shares with respect to the price. We will get an explicit
solution for the optimal execution problem, which extends the results given
by Alfonsi, Fruth and Schied~\cite{AFS2}.

\subsubsection{Volume impact reversion model}

When $f\equiv 1$, the deterministic cost function is simply given by
\begin{equation}\label{eq:cost_model_vol_const_deterministe}
C^V ({\bm \xi}, {\bm t}) =\sum^{N}_{n=0}\lambda(t_{n})\int^{D_{t_{n}+}}_{D_{t_{n}}}xf(x)dx=\sum^{N}_{i=0}\frac{\xi_{i}}{2}\left(\frac{\xi_{i}}{\lambda\left(t_{i}\right)}+2\frac{\sum_{j<i}e^{-\int^{t_{i}}_{t_{j}}\rho_{s}ds}\xi_{j}}{\lambda(t_{i})}\right),  
\end{equation}
which is a quadratic form: $C^V ({\bm \xi}, {\bm t}) =\frac{1}{2}{\bm \xi}^T
M^V {\bm \xi}$, with $M^V_{i,j}=\frac{\exp\left(-\left|\int_{t_i}^{t_j} \rho_s
    ds \right|\right)}{\lambda(t_i \vee t_j)}$, $0\le i,j \le N$.

\begin{theorem}
\label{thm:m_v_cte}
The quadratic form~\eqref{eq:cost_model_vol_const_deterministe} is positive
definite if and only if \begin{equation}
  \label{eq:condit_uniq_sol_vol_mod}
  a_i\tilde{a}_i<1,\forall i\in\left\{1,\ldots ,N\right\}.
\end{equation}

In this case, the optimal execution problem to liquidate $\vi$ shares on the
time-grid~\eqref{det_time_grid} admits a unique optimal
strategy ${\bm \xi}^\star$ which is  deterministic
and explicitly given by:
\begin{equation}
\label{opt_sol_vol_mod_const}
\begin{cases}
 \xi^\star_{0} &= -\frac{\vi}{K_V} \lambda(t_{0}) \frac{1-a_{1}}{1-a_1\tilde{a}_1} \\
  \xi^\star_{i} &= -\frac{\vi}{K_V}
 \lambda(t_i)\left[\frac{a_{i+1}}{1-a_{i+1}\tilde{a}_{i+1}}(\tilde{a}_{i+1}-1)+\frac{1-\tilde{a}_i}{1-a_i\tilde{a}_i}\right],
\ 1\le i \le N-1 \\ 
  \xi^\star_{N} &= -\frac{\vi}{K_V} \lambda(t_N) \frac{1-\tilde{a}_N
  }{1-a_N\tilde{a}_N },
\end{cases}
\end{equation}
where
$$K_V=\frac{\lambda\left(t_{0}\right)\left(1-2a_{1}\right)+\lambda\left(t_{1}\right)}{1-a_{1}\tilde{a}_1}+\sum^{N}_{i=2}\lambda(t_i)\frac{(1-\tilde{a}_i)^2}{1-a_i\tilde{a}_i}.$$
Its cost is given by $C^V ({\bm \xi}^\star, {\bm t})=\vi^2/(2K_V)$.
\end{theorem}

This theorem provides an explicit optimal strategy for the large trader. It
also gives explicit conditions that exclude or create PMS. First, let us
assume that
\begin{equation}
  \label{eq:asympt_limit_condit_vol_mod}
  \forall t\ge 0, \  2\rho_t +\eta_t\ge 0 .
\end{equation}
Then, for any time grid~\eqref{det_time_grid}, $a_i\tilde{a}_i \le 1$ and the
quadratic form~\eqref{eq:cost_model_vol_const_deterministe} is positive
semidefinite since it is a limit of positive definite quadratic forms. Thus,
the model is PMS free. Conversely let us assume that  $2\rho_{t_1} +\eta_{t_1}<0$ for
some $t_1\ge 0$. Let us consider the following round trip on the time grid
${\bf t}=(0,t_1,t_2)$ with $t_2>t_1$, where the large trader buys $x>0$ at time
$t_1$ and sells $x$ at time $t_2$. The cost of such a strategy is given by
\begin{equation}
  C^V ((0,x,-x), {\bm
    t})=\frac{x^{2}}{2\lambda(t_2)}\left(e^{\int^{t_2}_{t_1}\eta_{u}du}+1-2e^{-\int^{t_2}_{t_1}\rho_{u}du}\right)\underset{t_2
    \rightarrow t_1}{=}\frac{x^{2}}{2\lambda(t_1)}
  \left((2\rho_{t_1} +\eta_{t_1})(t_2-t_1)+o(t_2-t_1)\right)
\end{equation}
and is negative when $t_2$ is close enough to~$t_1$.
\begin{corollary}\label{cor_hs_v_bs}
In a block-shaped LOB,  model~$V$ does not
admit price manipulation in the sense of Huberman and Stanzl if and only if 
\eqref{eq:asympt_limit_condit_vol_mod}~holds.
\end{corollary}
Let us now discuss this result from the point of view of market makers. A
market maker that puts a significant orders may have an influence on $\rho_t$
and $\eta_t$ and can increase (resp. decrease) them by respectively adding
(resp. canceling) an order at a better price or at an existing limit order
price. What comes out from~\eqref{eq:asympt_limit_condit_vol_mod} is that no
PMS may arise if one adds limit orders, whatever the way of adding new
orders. Instead, PMS can occurs when canceling orders. A different conclusion
will hold in the price reversion model.

An analogous result to Corollary~\ref{cor_hs_v_bs} is stated in a recent paper
by Fruth, Schöneborn and Urusov~\cite{FSU} that has been
published while we were preparing this work. To be precise, results
in~\cite{FSU} are given for model~$P$ with a block-shaped LOB, and the optimal
execution strategy is obtained in a continuous time setting. As we will see in
the next paragraph, models~$V$ and~$P$ are mathematically equivalent when the
LOB shape is constant, even though they are different from a financial point
of view. By taking a regular time-grid
$t_i=\frac{iT}{N},i=0\ldots ,N$, and letting $N\rightarrow +\infty$, we get
back the optimal strategy in continuous time (that we still denote
by~$\xi^\star$, by a slight abuse of notations):
\begin{equation}
\label{opt_sol_vol_mod_const_cont}
\begin{cases}  
  {\mathbf{\xi}^{\star}_0}\underset{N\rightarrow +\infty}\longrightarrow
\xi^{\star}_0:=-\frac{\vi}{\lambda\left(T\right)+\int^{T}_{0}\frac{\rho^2_s\lambda(s)}{\eta_s+2\rho_{s}}ds}
\frac{\lambda(0)\rho_0}{2\rho_0+\eta_0 } \\
  \frac{\mathbf{\xi}^{\star}_{i_N}}{T/N}\underset{N\rightarrow
    +\infty}\longrightarrow \xi^{\star}_t:=-\frac{\vi}{\lambda\left(T\right)+\int^{T}_{0}\frac{\rho^2_s\lambda(s)}{\eta_s+2\rho_{s}}ds}
  \lambda(t) \left[\left(\frac{\rho_t}{2\rho_t+\eta_t}\right)'+\rho_t\left(\frac{\rho_t+\eta_t}{2\rho_t+\eta_t}\right)\right]
%\frac{\lambda(t) \left[\eta_t^2\rho_t +3\rho_t\eta_t+2\rho_t^3+\eta_t \rho_t' -\eta_t' \rho_t\right]}{\left(\eta_t+2\rho_t  \right)^2}
\text{, for } i_N\text{ such that }\frac{t_{i_N}}{N}\underset{\Delta t\rightarrow 0}\longrightarrow t \\
\mathbf{\xi}^{\star}_N \underset{N\rightarrow
  +\infty}\longrightarrow \xi^{\star}_T:=-\frac{\vi}{\lambda\left(T\right)+\int^{T}_{0}\frac{\rho^2_s\lambda(s)}{\eta_s+2\rho_{s}}ds}
\frac{\lambda(T)\left( \eta_T+ \rho_T \right)}{\eta_T+2\rho_T}.
\end{cases}
\end{equation}
The strategy $dX^\star_t={\xi^{\star}_0} \delta_0(dt)+\xi^{\star}_t dt+{\xi^{\star}_T}\delta_T(dt) $ with initial trade~${\xi^{\star}_0}$, continuous trading
${\xi^{\star}_t}$ on $[t,t+dt]$ for $t\in(0,T)$ and last trade
${\xi^{\star}_T}$ is indeed shown to be optimal in Fruth, Schöneborn and Urusov~\cite{FSU} among the continuous time
strategies with bounded variation. We will show here again this result for more general
LOB shape.  The optimal strategy has the following cost:
$$\frac{\vi^2}{2\left[\lambda\left(T\right)+\int^{T}_{0}\frac{\rho^2_s\lambda(s)}{2\rho_{s}+\eta_s}ds\right]}.$$
Besides, this provides a necessary and
sufficient condition to exclude transaction-triggered price manipulation.
\begin{corollary}\label{TTPM_vol_cte}
In a block-shaped LOB,  model~$V$ does not
admit transaction-triggered price manipulation  if and only if 
\begin{eqnarray}
&&\forall t\ge 0, \ \eta_t+ \rho_t \ge 0,\ \text{and }
\left(\frac{\rho_t}{2\rho_t+\eta_t}\right)'+\rho_t\left(\frac{\rho_t+\eta_t}{2\rho_t+\eta_t}\right)\ge 0.\label{cond_ttpms_v_bs}
\end{eqnarray}
\end{corollary}
The first condition comes from the last trade and implies~\eqref{eq:asympt_limit_condit_vol_mod} since
$\rho_t\ge 0$. It can be interpreted similarly as
condition~\eqref{eq:asympt_limit_condit_vol_mod} from market makers' point of
view. The second condition in~\eqref{cond_ttpms_v_bs} comes from
the intermediate trades and brings on the derivatives of~$\rho$ and $\eta$. It
is harder to get an intuitive idea of its meaning from a market maker's point
of view. Last, let us mention that
we can show that the optimal strategy on the discrete
time-grid~\eqref{det_time_grid} is made with nonnegative trades
if one has~\eqref{eq:condit_uniq_sol_vol_mod} and
\begin{equation}
  \label{eq:condit_t_t_manip_m_v}
  \frac{1-\tilde{a}_i}{1-a_i\tilde{a}_i}\geq a_{i+1}\frac{1-\tilde{a}_{i+1}}{1-a_{i+1}\tilde{a}_{i+1}},\;\;\forall i\in\left\{1,\ldots ,N-1\right\}\text{ and } \tilde{a}_N\leq 1.
\end{equation}
Condition~\eqref{cond_ttpms_v_bs} can be seen as the continuous time limit of condition~\eqref{eq:condit_t_t_manip_m_v}.

Let us give now an illustration of the optimal strategy with a time-varying
depth. We consider the case of a time-varying depth
\begin{equation*}\label{def_example_lambda}\lambda(t)=\lambda_0+\cos(2\pi t), \text{ with }\lambda_0>1,
\end{equation*}
which corresponds
to a one-day periodic function with high values at the beginning and at the
end of the day. We can show that $\eta_t \ge -\frac{2
  \pi}{\sqrt{\lambda_0^2-1}}$ and with a constant resilience $\rho$, there is
no PMS as soon as $2\rho-\frac{2
  \pi}{\sqrt{\lambda_0^2-1}} \ge 0$. Figure~\ref{OEblockshape} shows
the optimal execution strategy~\eqref{opt_sol_vol_mod_const} with a value
$\lambda_0$ that exclude PMS but allows TTPM. The optimal strategy to buy~$50$
shares consists in buying almost $95$~shares and selling $45$~shares, which
roughly treble the traded volume.
\begin{figure}[htbp]
 \centering
 \includegraphics[width=0.8\linewidth]{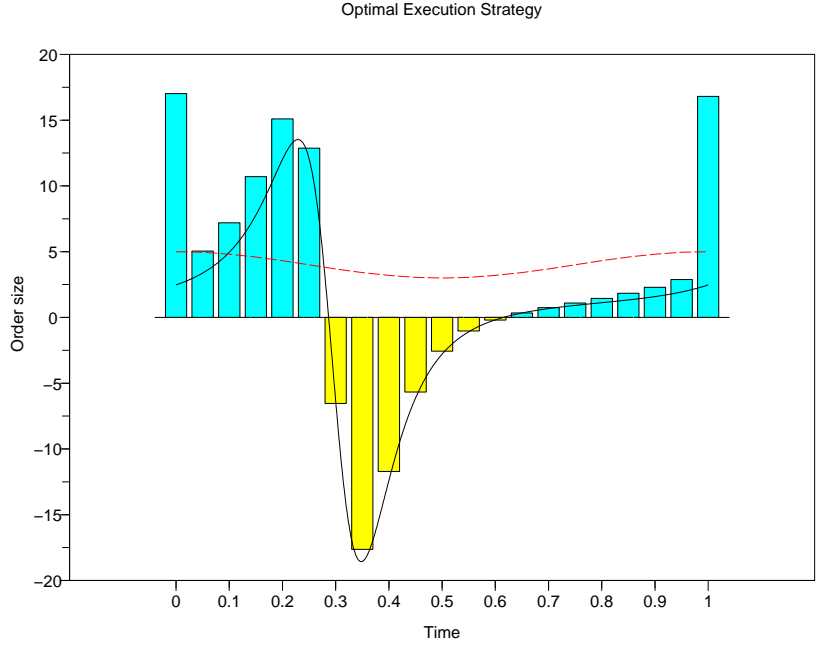} 
 \caption{Optimal execution strategy to buy $50$ shares on a regular time
   grid, with $N=20$, $\rho=1$, $\lambda(t)=4+\cos(2\pi t)$ (plotted in dashed
   line). In solid line is plotted the function $t\mapsto \left(\frac{\rho_t}{2\rho_t+\eta_t}\right)'+\rho_t\left(\frac{\rho_t+\eta_t}{2\rho_t+\eta_t}\right)$.}
\label{OEblockshape}
\end{figure}
 
\subsubsection{Price impact reversion model}\label{sec_model_P_bs}

When $f\equiv 1$, the deterministic cost function $\sum_{i=0}^N \int_{D_{t_i}}^{D_{t_i+}} \lambda(t_i)x f(x)dx  $ is given by
\begin{equation}\label{eq:cost_model_vol_const_deterministeP}
C^P ({\bm \xi}, {\bm t}) =\sum^{N}_{n=0}\lambda(t_{n})\int^{D_{t_{n}+}}_{D_{t_{n}}}xf(x)dx=\sum^{N}_{i=0}\frac{\xi_{i}}{2}\left(\frac{\xi_{i}}{\lambda(t_{i})}+2\sum_{j<i}e^{-\int^{t_{i}}_{t_{j}}\rho_{s}ds}\frac{\xi_{j}}{\lambda(t_{j})}\right).
\end{equation}
This is a quadratic form: $C^P ({\bm \xi}, {\bm t})=\frac{1}{2}{\bm \xi}^T M^P
{\bm \xi}$, with $M^P_{i,j} =
\frac{\exp\left(-\left|\int^{t_i}_{t_j}\rho_sds\right|\right)}{\lambda(t_i \wedge t_j)} $ for $0\le i, j
\le N$. When $f\equiv 1$, we get from~\eqref{dynamic_modelsVP} that model~$P$
is equivalent to model~$V$ with a
resilience~$\tilde{\rho}_t=\rho_t-\eta_t$. Another way to see that both models
are mathematically the same in the block-shape case is to reverse
the time and consider:
\begin{equation*}
\forall t\in[0,T], \ \hat{\rho}_t=\rho_{T-t},\;\hat{\lambda}(t)=\lambda(T-t)\text{ and }\hat{t}_{N-i}=T-t_i,\text{ for }0\leq i\leq N.
\end{equation*}
Then, we have
\begin{equation}
M^P_{i,j} = \frac{e^{-|\int^{t_i}_{t_j}\rho_sds|}}{\lambda(t_i\land t_j)} =\frac{e^{-|\int^{\hat{t}_{N-j}}_{\hat{t}_{N-i}}\hat{\rho}_sds|}}{\hat{\lambda}(\hat{t}_{N-i}\lor \hat{t}_{N-j})},
\end{equation}
and the optimal execution problem in Model~$P$ with resilience~$\rho$, LOB
depth~$\lambda(t)$ and time-grid~${\bm t}$ is the same as the optimal
execution problem in Model~$V$ with resilience $\hat{\rho}$, LOB
depth~$\hat{\lambda}(t)$ and time-grid~$\hat{\bm t}$. We immediately get the
following results.

\begin{theorem}
\label{thm:m_p_cte}
The quadratic form~\eqref{eq:cost_model_vol_const_deterministeP} is positive
definite if and only if 
\begin{equation}
\label{eq:condit_uniq_sol_price_mod}  
a_{i}\hat{a}_i<1,\forall i\in\left\{1,\ldots ,N\right\}
\end{equation}  
In this case, the optimal execution problem to liquidate $\vi$ shares on the
time-grid~\eqref{det_time_grid} admits a unique optimal
strategy ${\bm \xi}^\star$ which is  deterministic
and explicitly given by:
\begin{equation}
\label{opt_sol_price_mod_const}
\begin{cases}
\mathbf{\xi}^{\star}_{0}=-\frac{\vi}{K_P} \lambda(t_{0})  \frac{1-\hat{a}_{1}}{1-a_1\hat{a}_1}. \\
\mathbf{\xi}^{\star}_{i}=-\frac{\vi}{K_P}\lambda(t_i)\left[
\frac{a_i}{1-a_i \hat{a}_i}(\hat{a}_i-1)+\frac{1-\hat{a}_{i+1}}{1-a_{i+1}\hat{a}_{i+1}}\right],\ 1\le i \le N-1\\
%  \frac{(1-a_{i})}{1-a_i\hat{a}_i}-\frac{\hat{a}_{i+1}}{1-a_{i+1}\hat{a}_{i+1}}(1-a_{i+1})\right]
%, i=1,\ldots ,N-1 \\
\mathbf{\xi}^{\star}_{N}=-\frac{\vi}{K_P}\lambda(t_{N})\frac{1-a_{N}}{1-a_N\hat{a}_N}
\end{cases}
\end{equation}
where 
$$K_P=\frac{\lambda(t_N)(1-2a_N)+\lambda(t_{N-1})}{1-a_N\hat{a}_N}+\sum^{N-2}_{i=0}\lambda(t_i)\frac{(1-
  \hat{a}_{i+1})^2}{1-a_{i+1}
  \hat{a}_{i+1}}.$$
Its cost is given by $C^P ({\bm \xi}^\star, {\bm t})=\vi^2/(2K_P)$.
\end{theorem}
By taking a regular time-grid
$t_i=\frac{iT}{N},i=0\ldots ,N$, and letting $N\rightarrow +\infty$, we get
the optimal strategy in continuous time:
\begin{equation}
\begin{cases}
  {\xi^{\star}_0}\underset{N \rightarrow \infty}\longrightarrow \xi^{\star}_0:= -\frac{\vi}{\lambda(0)+\int^T_{0}\frac{\rho^2_s\lambda(s)}{2\rho_s-\eta_s}ds}\lambda(0) \frac{\rho_0-\eta_0}{2\rho_0-\eta_0} \\
  \frac{\xi^{\star}_{i{_N}}}{T/N}\underset{N \rightarrow
    \infty}\longrightarrow \xi^{\star}_t:=
-\frac{\vi}{\lambda(0)+\int^T_{0}\frac{\rho^2_s\lambda(s)}{2\rho_s-\eta_s}ds}\lambda(t)\left[\left(\frac{\rho_t-\eta_t}{2\rho_t-\eta_t}\right)'+\rho_t\left(\frac{\rho_t-\eta_t}{2\rho_t-\eta_t}\right)\right]\text{, for } i_N\text{ such that }\frac{Ti_N}{N}\underset{\Delta t\rightarrow 0}\longrightarrow t \\
  \xi^{\star}_N\underset{N \rightarrow \infty}\longrightarrow
  \xi^{\star}_T:=-\frac{\vi}{\lambda(0)+\int^T_{0}\frac{\rho^2_s\lambda(s)}{2\rho_s-\eta_s}ds}\lambda(T)
  \frac{\rho_T}{2\rho_T-\eta_T }.
\end{cases}
\end{equation}
The strategy with initial trade~${\xi^{\star}_0}$, continuous trading
${\xi^{\star}_t}$ on $[t,t+dt]$ for $t\in(0,T)$ and last trade
${\xi^{\star}_T}$ is shown to be optimal in Fruth, Schöneborn and Urusov~\cite{FSU} among the continuous time
strategies with bounded variation, and  has the following cost:
$$\frac{\vi^2}{2\left[\lambda\left(0\right)+\int^T_{0}\frac{\rho^2_s\lambda(s)}{2\rho_s-\eta_s}ds
  \right]}.$$

\begin{corollary}\label{cor_hs_p_bs}
In a block-shaped LOB,  model~$P$ does not
admit price manipulation in the sense of Huberman and Stanzl if and only if 
 \begin{equation}
  \label{eq:asympt_limit_condit_price_mod}
\forall t\ge 0, \   2\rho_t - \eta_t \ge 0.
\end{equation}
It does not
admit transaction-triggered price manipulation  if and only if 
\begin{eqnarray}
&&\forall t\ge 0, \  \rho_t-\eta_t \ge 0,\ \text{and }
\left(\frac{\rho_t-\eta_t}{2\rho_t-\eta_t}\right)'+\rho_t\left(\frac{\rho_t-\eta_t}{2\rho_t-\eta_t}\right)\ge 0.\label{cond_ttpms_p_bs}
\end{eqnarray}
\end{corollary}
The first condition in~\eqref{cond_ttpms_p_bs} comes from the initial trade
while the second comes from intermediate trades. From market makers' point of
view, \eqref{eq:asympt_limit_condit_price_mod} and the first condition
in~\eqref{cond_ttpms_p_bs} give different conclusions from model~$V$. A
significant market maker will not create manipulation strategy if he puts
orders at a better price (which increases $\rho$) or cancels orders at existing
prices (which decreases $\eta$). Instead, he may create manipulation
strategies if he piles orders at existing prices, or if he cancels orders that
are among the best offers. The second condition of~\eqref{cond_ttpms_p_bs}
brings on the dynamics of~$\eta$ and~$\rho$ and it is more difficult to give its
heuristic meaning in terms of trading. Last, let us mention that the optimal
strategy in discrete time given by Theorem~\ref{thm:m_p_cte} is made only with trades
of same sign if, and only if, one has~\eqref{eq:condit_uniq_sol_price_mod} and  
\begin{equation}
  \label{eq:condit_t_t_manip_m_p}
\frac{1-\hat{a}_{i+1}}{1-a_{i+1}\hat{a}_{i+1}} \ge a_i \frac{1-\hat{a}_i}{1-a_i\hat{a}_i}
,\;\;\forall i\in\left\{1,\ldots ,N-1\right\}  \text{ and }\hat{a}_1<1.
\end{equation}

\subsection{Results for general LOB shape}
\label{results_general_case}

We extend in this section the results obtained on the optimal execution for
block-shaped LOB to more general shapes. In particular, the necessary and
sufficient conditions that we have obtained to exclude TTPM
(namely~\eqref{cond_ttpms_v_bs} for model~$V$ and~\eqref{cond_ttpms_p_bs} for
model~$P$) are still sufficient conditions to exclude TTPM for a wider class
of shape functions.  From a mathematical point of view, the approach is the
same. We first characterize the optimal strategy on a discrete time grid, by
using Lagrange multipliers. Then, one can guess the optimal continuous time
strategy, and we prove its optimality by a verification argument.

%% In this section, we consider that the distribution of orders in the limit order book can be represented by the product of two functions $\lambda(t)f(\xi)$ where $f:\mathbb{R}\mapsto[0,\infty)$ and $\lambda:[0,T]\mapsto(0,\infty)$ are continuous functions.

%% This time-dependent framework generalizes the shape of the limit order book studied in \cite{alfonsi-schied-slynko-2011}. 

%% We show under certain hypothesis the existence of an optimal strategy for the large trader in the case of the volume and the price impact model.

%% Unlike the case of a constant shape function for the limit order book, there is no relationship between the two models and therefore,  we describe in detail each of them.

%% We recall that given a sequence of trading times $\mathbf{t}=\left(t_0,\ldots ,l_N\right)$, we have to find
%% the strategy $\mathbf{\xi}=\left(\xi_0,\ldots ,\xi_N\right)$ that minimizes the average cost $\mathbf{C_{\mathbf{t}}}(\xi)$ determined by the
%% equations $\eqref{eq:cost_one_trade}$ and $\eqref{eq:average_cost}$.

\subsubsection{Volume impact reversion model}
\label{results_general_case_vol}

We first introduce the  following assumption that will be useful to study the
optimal discrete strategy.

\begin{assumption}
  \label{assumption_mod_vol}
%In the framework of the volume impact reversion model, we assume that:
\begin{enumerate}
%% \item The limit order book has infinite depth, that is, the function F is unbounded in the sense that
%%   \begin{equation*}
%%   \lim_{x\uparrow\infty}F(x)=\infty\text{ and }\lim_{x\downarrow -\infty}F(x)=-\infty
%%   \end{equation*}
\item The shape function $f$ satisfies the following condition:
  \begin{equation*}
   f \text{ is nondecreasing on }\mathbb{R_{-}}\text{ and nonincreasing on }\mathbb{R_{+}}
  \end{equation*}
\item $\forall t\ge 0, \rho_t+\eta_t \ge 0$. 
%  The quantity $\tilde{a}_{i}:=\frac{\lambda(t_{i-1})}{\lambda(t_{i})}a_{i}\leq 1$ for all $i\in\left\{1,\ldots ,N\right\}$
\end{enumerate}  
\end{assumption}
We remark that when the LOB shape does not evolve in time ($\eta_t=0$), the second
condition is satisfied  and we get back the assumption made
in Alfonsi, Fruth and Schied~\cite{AFS}. We define
\begin{equation}\label{def_hvi}
  x\in \R, \ 
h_{V,i}(x)=\frac{F^{-1}(x)-a_{i}F^{-1}\left(\tilde{a}_{i}x\right)}{1-a_i},
1\le i\le N. 
\end{equation}

\begin{theorem}
\label{thm_vol_mod}
Under Assumption~\ref{assumption_mod_vol}, the cost function $C^V ({\bm \xi},
{\bm t})$ is nonnegative, % =\sum^{N}_{n=0}\lambda(t_{n})\int^{D_{t_{n}+}}_{D_{t_{n}}}xf(x)dx \ge0$.
and there is a unique optimal
execution strategy ${\bm \xi}^\star$ that minimizes~$C^V$ over $\{ {\bm
  \xi}\in \R^{N+1}, \sum_{i=0}^N \xi_i=-\vi \}$. This strategy is given as
follows. The following equation
$$\sum_{i=1}^N\lambda(t_{i-1})(1-a_i)h^{-1}_{V,i}(\nu)
+\lambda(t_N)F(\nu)=-\vi $$
has a unique solution $\nu \in \R$, and 
\begin{eqnarray*}
\xi^\star_{0}&=&\lambda(t_{0})h^{-1}_{V,1}\left(\nu\right),  \\
\xi^\star_{i}&=&\lambda(t_{i})(h^{-1}_{V,i+1}\left(\nu\right)-\tilde{a}_{i}h^{-1}_{V,i}\left(\nu\right)),\
1\le i\le N-1, \\
\xi^\star_{N}&=&\lambda(t_{N})F(\nu)-\lambda(t_{N-1})a_{N}h^{-1}_{V,N}\left(\nu\right).
\end{eqnarray*}
The first and the last trade have the same sign as~$-\vi$. Besides, if the
following condition holds
\begin{equation}
  \label{eq:condit_manip_price_vol_mod}
\frac{1}{\tilde{a}_{i}}\frac{1-\tilde{a}_{i}}{1-a_{i}}\geq\frac{1-\tilde{a}_{i+1}}{1-a_{i+1}}    ,
\end{equation}
the intermediate trades $\xi^\star_{i},\ 1\le i\le N-1,$ have also the same sign as~$-\vi$.
\end{theorem}
This theorem extends the results of~\cite{AFS}, where $\lambda$ is assumed to
be constant. In that case, \eqref{eq:condit_manip_price_vol_mod}~is satisfied
and all the trades have the same
sign. Condition~\eqref{eq:condit_manip_price_vol_mod} is interesting since it
does not depend on the shape function, but it is more restrictive than the
condition~\eqref{eq:condit_t_t_manip_m_v} for the block-shape case (see
Lemma~\ref{lemma_implication} for
$\eqref{eq:condit_manip_price_vol_mod}\implies
\eqref{eq:condit_t_t_manip_m_v}$). In fact, the continuous time formulation is
more convenient to analyze the sign of the trades. Under Assumption~\ref{assumption_mod_vol}, we will show that no
transaction-triggered price manipulation can occur with the same
condition~\eqref{cond_ttpms_v_bs} as for the block-shape case. 

 %% This is basically due to the argument of
%% the proof that relies on a verification argument, while the proof that we give
%% here in the discrete case relies on th

When stating the optimal continuous-time strategy, we slightly relax
Assumption~\ref{assumption_mod_vol}.  This is basically due to the argument of
the proof that relies on a verification argument. Instead, our proof in the
discrete case relies on Lagrange multipliers which requires to show first
that the cost function has a minimum, and we use~$\rho_t+\eta_t \ge
0$ for that. We introduce the following function
\begin{equation}
  h_{V,t}(x)=F^{-1}(x)+\frac{\eta_t+\rho_t}{\rho_t} \frac{x}{f(F^{-1}(x))}.
\end{equation}
We will show that no PMS exists and that there is a unique optimal strategy if
these functions for $t\in [0,T]$ are bijective on~$\R$ with a positive
derivative. If Assumption~\ref{assumption_mod_vol} holds, this condition
is automatically satisfied.

\begin{theorem}\label{thm_vol_mod_cont}
 Let  $f\in \cC^1(\R)$. We assume that for $t\in[0,T]$, $h_{V,t}$ is
 bijective on~$\R$, such that
 $h_{V,t}'>0$. Then, the cost function $C^V (X)$ is nonnegative,
and there is a unique optimal
admissible strategy $X^\star$ that minimizes~$C^V$. This strategy is given as
follows. The equation
\begin{equation}\label{def_nu_vi_vol}
  \int_0^T\lambda(t)\rho_t h_{V,t}^{-1}(\nu)dt+\lambda(T)F(\nu)=-\vi
\end{equation}
has a unique solution $\nu \in \R$ and we
set~$\zeta_t=h_{V,t}^{-1}(\nu)$. The strategy $dX^\star_t=\xi^\star_{0} \delta_0(dt)+\xi^\star_{t}dt+\xi^\star_{T} \delta_T(dt)
$ with 
\begin{eqnarray*}
\xi^\star_{0}&=&\lambda(0) \zeta_0,  \\
\xi^\star_t&=&\lambda(t)\left[\frac{d \zeta_t}{dt}+(\rho_t+\eta_t)\zeta_t \right], \\
\xi^\star_{T}&=&\lambda(T)(F(\nu)-\zeta_T),
\end{eqnarray*}
is optimal. The initial trade~$\xi^\star_{0}$ has the same sign
as~$-\vi$. 
\end{theorem}

Thus, a sufficient condition to exclude price manipulation strategies is to
assume that $h_{V,t}$ is bijective with $h_{V,t}'>0$. We have a partial
reciprocal result: there are PMS as soon as $h_{V,t_1}'(0)<0$ for some~$t_1\ge
0$. Indeed, in this case we consider the following round trip on the time grid
${\bf t}=(0,t_1,t_2)$ with $t_2>t_1$, where the large trader buys $x>0$ at
time~$t_1$ and sells $x$ at time~$t_2$. The cost of such a strategy is given by
\begin{eqnarray*}
  C^V ((0,x,-x),{\bm
    t})&=&\lambda(t_1)G\left(\frac{x}{\lambda(t_1)}\right)+\lambda(t_2)\left(G\left(\frac{x(e^{-\int_{t_1}^{t_2}
        \rho_s ds}-1)}{\lambda(t_2)}\right)-G\left(\frac{xe^{-\int_{t_1}^{t_2}
        \rho_s ds}}{\lambda(t_2)}\right)\right)\\
&=&\lambda(t_1)\left( -\eta_{t_1}G\left(\frac{x}{\lambda(t_1)}\right)+(\rho_{t_1}+\eta_{t_1})\frac{x}{\lambda(t_1)}F^{-1}\left(\frac{x}{\lambda(t_1)}\right) \right) (t_2-t_1) +o(t_2-t_1).\\
\end{eqnarray*}
The derivative of $x\mapsto -\eta_{t_1}G(x)+(\rho_{t_1}+\eta_{t_1})xF^{-1}(x)$
is $\rho_{t_1} h_{V,t_1}(x)$, which has the opposite sign of~$x$ near~$0$ since $
h_{V,t_1}(0)=0$ and $h_{V,t_1}(0)<0$ by assumption. Thus, we have $C^V
((0,x,-x),{\bf t})<0$ for $x$ and $t_2-t_1$ small enough.

Now, let us focus on the sign of the trades given by the optimal
strategy. Without further hypothesis, the condition $\xi^\star_t\ge 0$ typically
involves the shape function~$f$. However, under
Assumption~\ref{assumption_mod_vol}, we can show that
transaction-triggered strategy are excluded under the same assumption as for
the block-shape case.

\begin{corollary}\label{cor_vol_mod_cont}
Let $f\in \cC^1$. Under Assumption~\ref{assumption_mod_vol}, the function
$h_{V,t}$ is $\cC^1(\R)$, bijective on~$\R$, and such that $h_{V,t}'>0$. Thus, the
result of Theorem~\ref{thm_vol_mod_cont} holds and the last trade 
$\xi^\star_{T}$ has the same sign as $-\vi$.

 Besides, if~\eqref{cond_ttpms_v_bs} also holds, $\xi^\star_t$  has the
same sign as~$-\vi$ for any $0<t<T$, which excludes TTPM.
\end{corollary}

Let us now focus on the example of a power-law shape: we assume that
$$f(x)=|x|^\gamma, \ \gamma>-1. $$
In this case, $F(x)=\sgn(x)\frac{|x|^{\gamma+1}}{\gamma+1}$ is well-defined
and satisfies~\eqref{infinite_d_ordres}. We have
$F^{-1}(x)=\sgn(x)(\gamma+1)^{\frac{1}{\gamma+1}}|x|^{\frac{1}{\gamma+1}}$ and
$h_{V,t}(x)=\sgn(x)(\gamma+1)^{\frac{1}{\gamma+1}}|x|^{\frac{1}{\gamma+1}}\left(\frac{\rho_t(2+\gamma)+\eta_t}{\rho_t(1+\gamma)}
\right)$. Thus, $h_{V,t}$ is bijective and increasing
 if, and only if:
$$\rho_t(2+\gamma)+\eta_t>0. $$
In this case, we have
$$ h_{V,t}^{-1}(x)=\frac{1}{\gamma+1} K_t(\gamma) \sgn(x) |x|^{\gamma+1} \text{ with }
K_t(\gamma)= \left(\frac{\rho_t(1+\gamma)}{\rho_t(2+\gamma)+\eta_t} \right)^{1+\gamma}.$$
In this case, we have by Theorem~\ref{thm_vol_mod_cont} that
\begin{equation}\label{powerlaw_opt_strat}
\begin{cases}  \xi^\star_0=\frac{-\vi}{\int_0^T \lambda(t) \rho_t K_t(\gamma) dt
    +\lambda(T) }\lambda(0)K_0(\gamma), \\
  \xi^\star_t= \frac{-\vi}{\int_0^T \lambda(t) \rho_t K_t(\gamma) dt
    +\lambda(T) } \lambda(t) \left[\frac{dK_t(\gamma)}{dt} +(\rho_t+\eta_t)K_t(\gamma)\right]\\
  \xi^\star_T=\frac{-\vi}{\int_0^T \lambda(t) \rho_t K_t(\gamma) dt
    +\lambda(T) } \lambda(T)(1-K_T(\gamma))
\end{cases}
\end{equation}
is the unique optimal strategy. For $\gamma=0$, we get
back~\eqref{opt_sol_vol_mod_const_cont}. If we only assume that
$\rho_t(2+\gamma)+\eta_t \ge 0 $, we still have $C^V(X)\ge 0$ for any
admissible strategy~$X$. The cost~$C^V(X)$ is indeed continuous with respect
to the resilience, and is the limit of the cost associated to resilience
$\rho_t+\varepsilon$, $\varepsilon \downarrow 0$.
On the contrary, if
$\rho_t(2+\gamma)+\eta_t < 0 $, we have $h'_{V,t}(0)<0$  and there is a PMS as
explained above.
\begin{corollary}\label{cor_powerlaw_v}
When $f(x)=|x|^\gamma$,  model~$V$ does not
admit PMS if, and only if 
 \begin{equation*}
\forall t\ge 0, \  \rho_t(2+\gamma)+\eta_t \ge 0.
\end{equation*}
It does not
admit transaction-triggered price manipulation  if and only if 
\begin{eqnarray*}
&&\forall t\ge 0, \  \rho_t+\eta_t \ge 0,\ \text{and }
\left(\frac{\rho_t(1+\gamma)}{\rho_t(2+\gamma)+\eta_t} \right)'+\rho_t\left(\frac{\rho_t+\eta_t}{\rho_t(2+\gamma)+\eta_t }\right)\ge 0.
\end{eqnarray*}
\end{corollary}
\noindent These conditions comes respectively from the nonnegativity of the last and
intermediate trades. For given functions $\rho_t$ and $\eta_t$, the no PMS
condition will be satisfied for $t\in[0,T]$ when $\gamma$ is large
enough. This can be explained heuristically. When $\gamma$~increases, limit
orders become rare close to~$S^0_t$ and dense away from~$S^0_t$, which creates
some bid-ask spread. One has then to pay to get liquidity, and round trips
have a positive cost. Instead, when $\gamma$ is close to~$-1$ it is rather
cheap to consume limit orders, which may facilitate PMS. In
Figure~\ref{OEPowerlaw}, we have plotted the optimal strategy for
$\gamma=-0.3$ and $\gamma=1$ with the same parameters as in Figure~\ref{OEblockshape} for
the Block shape case. We can check that the no PMS condition is satisfied in
both cases.

\begin{figure}[htbp]
\begin{minipage}[c]{0.5\textwidth}
\centering 
\includegraphics[width=3.5in]{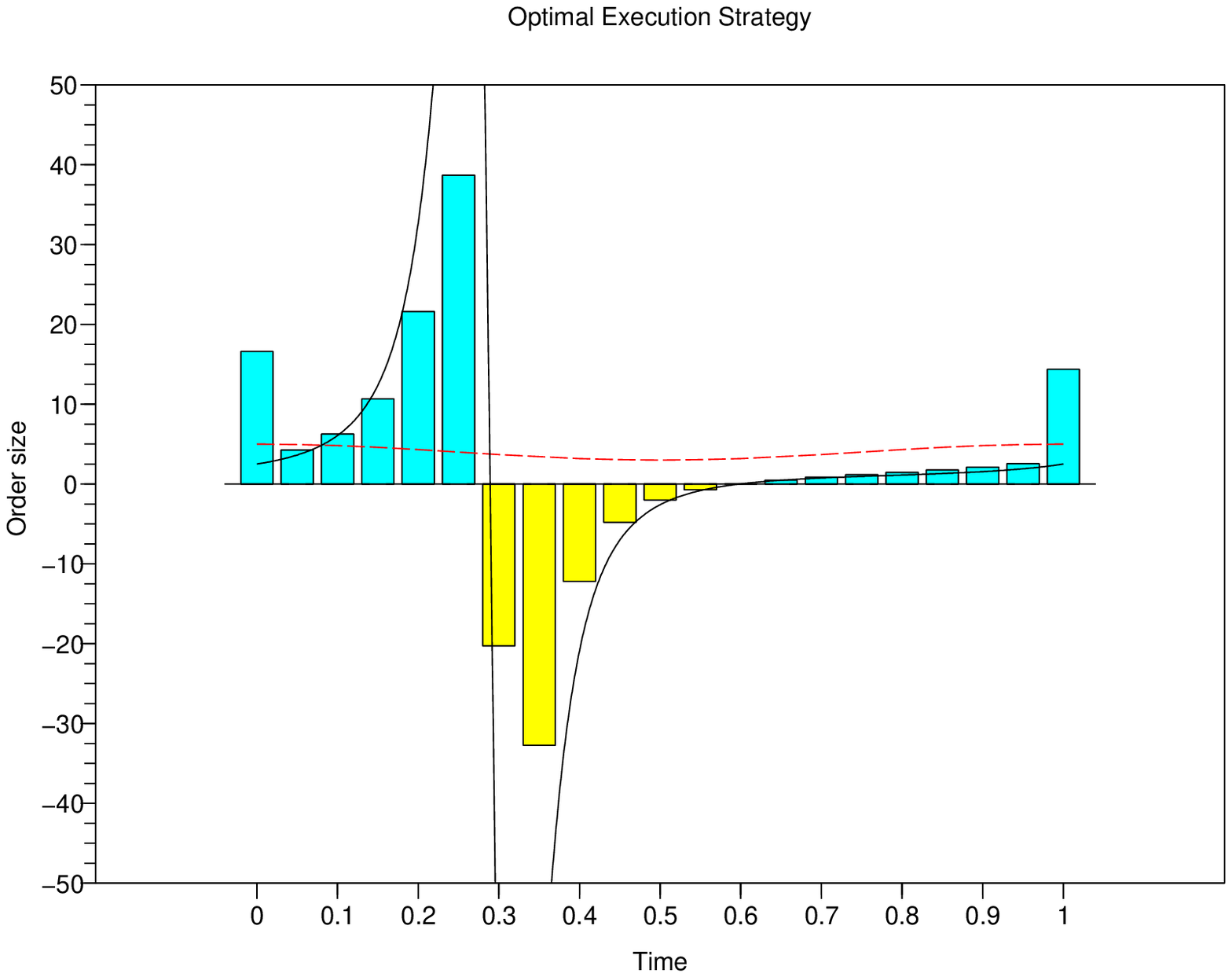}
\end{minipage}
\begin{minipage}[c]{0.5\textwidth}
\centering 
\includegraphics[width=3.5in]{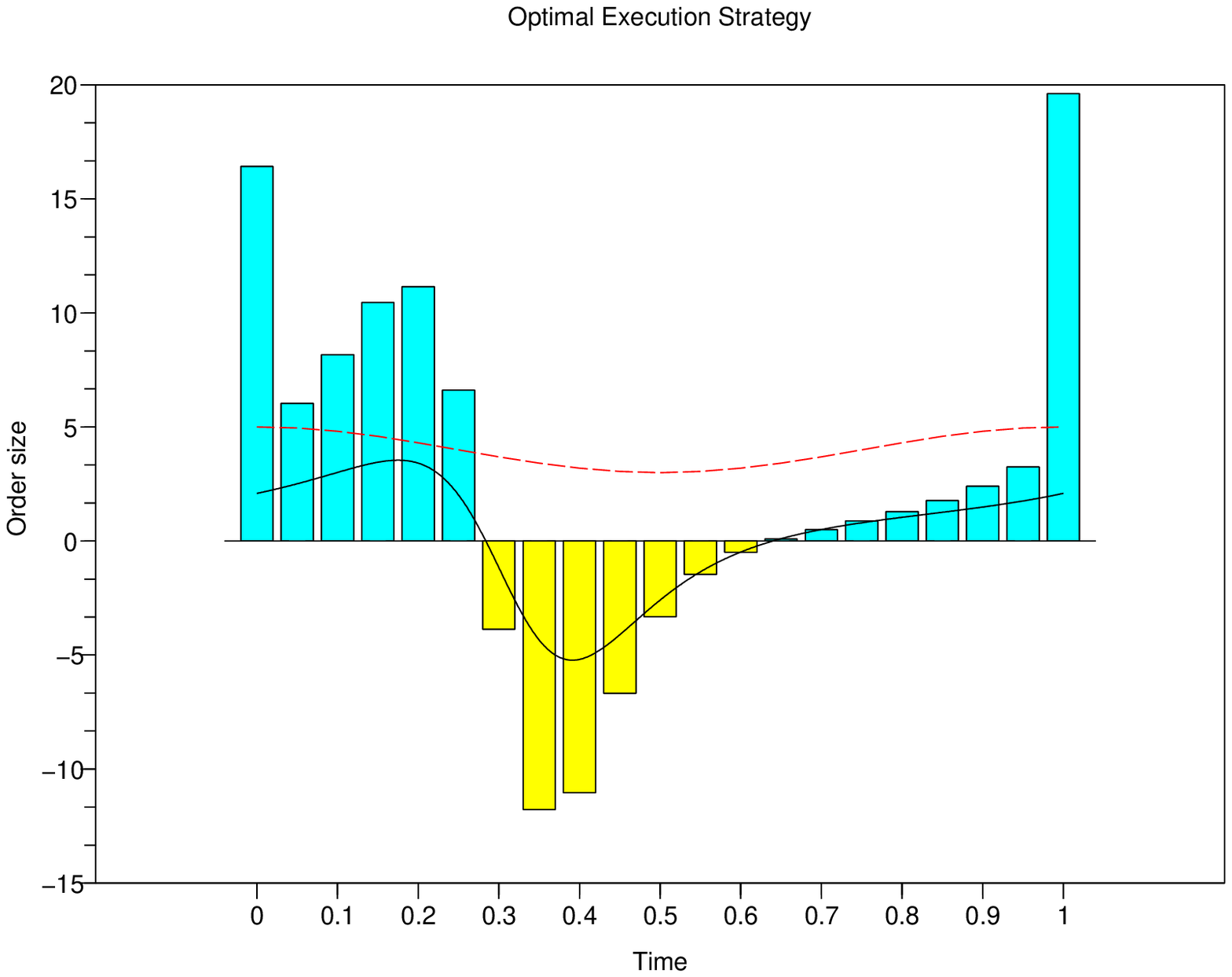}
\end{minipage}
 \caption{Optimal execution strategy to buy $50$ shares on a regular time
   grid, with $N=20$, $\rho=1$, $\lambda(t)=4+\cos(2\pi t)$ (plotted in dashed
   line) and $\gamma=-0.3$ (left) or $\gamma=1$ (right). In solid line is
   plotted the function $t\mapsto
   \left(\frac{\rho_t(1+\gamma)}{\rho_t(2+\gamma)+\eta_t}
   \right)'+\rho_t\left(\frac{\rho_t+\eta_t}{\rho_t(2+\gamma)+\eta_t }\right)
   $ (this function is well-defined but out of the graph for $\gamma=-0.3$).}
\label{OEPowerlaw}
\end{figure}

\subsubsection{Price impact reversion model}
\label{results_general_case_price}

The results that we present for model~$P$ are similar to the one obtained for
model~$V$. We first solve the optimal execution problem in discrete
time. From its explicit solution, we then calculate its continuous time limit
and check by a verification argument that it is indeed optimal. Doing so, we get
sufficient conditions to exclude PMS and TTPM. In particular, 
condition~\eqref{cond_ttpms_p_bs} that excludes PMS and TTPM for
block-shape LOB also excludes PMS and TTPM for a general LOB shape satisfying Assumption~\ref{assumption_mod_price} below.

To study the optimal discrete strategy, we will work under the following  assumption.
\begin{assumption}
  \label{assumption_mod_price}
\begin{enumerate}
\item The shape function $f$ is~$\cC^1$ and satisfies the following condition:
  \begin{equation*}
   f \text{ is nonincreasing on }\mathbb{R_{-}}\text{ and nondecreasing on }\mathbb{R_{+}}
  \end{equation*}
\item $\forall t\ge 0, \rho_t-\eta_t > 0$. 
\item $x\mapsto x\frac{f'(x)}{f(x)}$ is nondecreasing on $\mathbb{R}_{-}$,
  nonincreasing on $\mathbb{R}_{+}$.   
\end{enumerate}
\end{assumption}
The monotonicity assumption made here is the opposite to the one made in
Assumption~\ref{assumption_mod_vol} for model~$V$. This choice is different
from the one made in Alfonsi et al.~\cite{AFS,AS}. It is in fact more
tractable from a mathematical point of view, especially here with a
time-varying LOB.
\begin{theorem}
\label{thm_price_mod}
Under Assumption~\ref{assumption_mod_price}, the cost function $C^P ({\bm \xi},
{\bm t})$ is nonnegative, % =\sum^{N}_{n=0}\lambda(t_{n})\int^{D_{t_{n}+}}_{D_{t_{n}}}xf(x)dx \ge0$.
and there is a unique optimal
execution strategy ${\bm \xi}^\star$ that minimizes~$C^P$ over $\{ {\bm
  \xi}\in \R^{N+1}, \sum_{i=0}^N \xi_i=-\vi \}$. This strategy is given as
follows. The following equation
$$\sum_{i=1}^N\lambda(t_{i-1})\left[F\left(\frac{h^{-1}_{P,i}(\nu)}{a_{i}}\right)-\frac{\lambda(t_i)}{\lambda(t_{i-1})}F(h^{-1}_{P,i}(\nu))\right]
+\lambda(t_N)F(\nu)=-\vi $$
has a unique solution $\nu \in \R$, and 
\begin{eqnarray*}
\xi^\star_{0}&=&\lambda(t_{0})F\left(\frac{h^{-1}_{P,1}\left(\nu\right)}{a_1}\right),  \\
\xi^\star_{i}&=&\lambda(t_{i})\left[F\left(\frac{h^{-1}_{P,i+1}\left(\nu\right)}{a_{i+1}}\right)-F(h^{-1}_{P,i}(\nu)) \right]   ,\
1\le i\le N-1, \\
\xi^\star_{N}&=&\lambda(t_{N})[F(\nu)-F(h^{-1}_{P,N}\left(\nu\right))].
\end{eqnarray*}
The first and the last trade have the same sign as~$-\vi$.
\end{theorem}

%% This result does not only give an optimal strategy for the large trader becuase we can deduce conditions for the viability of the market. 

%% \begin{corollary}
%% Let us introduce the constant $\nu>0$ and the functions: $h_{P,i}(\xi)=\xi\frac{\left[\frac{1}{a_i}f(\frac{\xi}{a_i})-\hat{a}_if(\xi)\right]}{f\left(\frac{\xi}{a_i}\right)-\hat{a}_if(\xi)}$.
%% Under the same conditions that in Theorem \ref{thm_price_mod}, the price impact reversion model does not admit transaction-triggered price manipulation if and only if the following condition is verified:
%% \begin{equation}
%%   \label{eq:condit_thm_price_mod}
%%   \begin{cases}
%% \frac{h^{-1}_{P,i+1}(\nu)}{a_{i+1}}\geq h^{-1}_{P,i}(\nu)\text{ for }
%% i=1,\ldots ,N-1 \\
%% \text{ and } \nu\geq h^{-1}_{P,N}(\nu)
%%   \end{cases} 
%% \end{equation}
%% \end{corollary}

%% We show in the Appendix \ref{app_price_impact_general} that this last condition implies $\eqref{eq:condit_t_t_manip_m_p}$ when we consider a constant shape function for the distribution of orders in the limit order book.

We now state the corresponding result in continuous time and set:
\begin{equation}
x \in \R, \  h_{P,t}(x)=x \left[1+\frac{\rho_t} {\rho_t\left(1+\frac{xf'(x)}{f(x)}\right)-\eta_t} \right].
\end{equation}

\begin{theorem}\label{thm_price_mod_cont}
 Let  $f\in \cC^2(\R)$. We assume that one of the two following conditions holds.
 \begin{enumerate}[(i)]
 \item For $t\in[0,T]$,
 $\rho_t\left(1+\frac{xf'(x)}{f(x)}\right)-\eta_t>0$ for any $x\in \R$ and $h_{P,t}$ is  bijective on~$\R$, such that
 $h_{P,t}'(x)>0$, $dx$-a.e.
\item For $t\in[0,T]$,
 $\rho_t\left(1+\frac{xf'(x)}{f(x)}\right)-\eta_t<0$ and
 $\rho_t\left(2+\frac{xf'(x)}{f(x)}\right)-\eta_t>0$ for any $x\in \R$, and $h_{P,t}$ is  bijective on~$\R$, such that
 $h_{P,t}'(x)<0$, $dx$-a.e.
 \end{enumerate}
 Then, the cost function $C^P (X)$ is nonnegative,
and there is a unique optimal
admissible strategy $X^\star$ that minimizes~$C^P$. This strategy is given as
follows. The equation
\begin{equation}\label{def_nu_vi_pr}
  \int_0^T\lambda(t)[\rho_t h_{P,t}^{-1}(\nu) f(h_{P,t}^{-1}(\nu))-\eta_tF(h_{P,t}^{-1}(\nu))]dt+\lambda(T)F(\nu)=-\vi
\end{equation}
has a unique solution $\nu \in \R$ and we
set~$\zeta_t=h_{P,t}^{-1}(\nu)$. The strategy $dX^\star_t=\xi^\star_{0} \delta_0(dt)+\xi^\star_{t}dt+\xi^\star_{T} \delta_T(dt)
$ with 
\begin{eqnarray*}
\xi^\star_{0}&=&\lambda(0) F(\zeta_0),  \\
\xi^\star_t&=&\lambda(t)f(\zeta_t)\left[\frac{d \zeta_t}{dt}+\rho_t\zeta_t \right], \\
\xi^\star_{T}&=&\lambda(T)(F(\nu)-F(\zeta_T)),
\end{eqnarray*}
is optimal. The initial trade~$\xi^\star_{0}$ has the same sign
as~$-\vi$. 
\end{theorem}
In particular, there is no PMS in model~$P$ as soon as Assumptions~$(i)$
or~$(ii)$ hold. Conversely, let us assume that
$ \rho_{t_1}\left(2+\frac{xf'(x)}{f(x)}\right)-\eta_{t_1}<0, $ when $x$ belongs to a
neighbourhood of~$0$ for some $t_1\ge 0$. Then, we set
${\bf t}=(0,t_1,t_2)$ with $t_2>t_1$, and consider that the large trader buys $x>0$ at
time~$t_1$ and sells $x$ at time~$t_2$. The cost of such a round trip is 
\begin{eqnarray*}
&&  C^P ((0,x,-x),{\bm
    t})\\
  &=&\lambda(t_1)G\left(\frac{x}{\lambda(t_1)}\right)+\lambda(t_2)\left[G\left(F\left(e^{-\int_{t_1}^{t_2}
        \rho_s ds} F^{-1}\left( \frac{x}{\lambda(t_2)} \right)\right)-\frac{x}{\lambda(t_2)}
     \right)-\tilde{F}\left(e^{-\int_{t_1}^{t_2}
        \rho_s ds} F^{-1}\left( \frac{x}{\lambda(t_2)} \right)\right) \right]\\
&=&\lambda(t_1)\left[
  -\eta_{t_1} \tilde{F}\left(F^{-1}\left(\frac{x}{\lambda(t_1)}\right)\right)+\rho_{t_1} F^{-1}\left(\frac{x}{\lambda(t_1)}
  \right)^2 f\left(F^{-1}\left(\frac{x}{\lambda(t_1)}\right) \right) \right] (t_2-t_1) +o(t_2-t_1).\\
\end{eqnarray*}
The derivative of $x\mapsto -\eta_{t_1} \tilde{F}(x)+\rho_{t_1} x^2f(x)$ is
$xf(x)\left(\rho_{t_1}\left(2+\frac{xf'(x)}{f(x)}\right)-\eta_{t_1}\right)$ and has the
opposite sign of~$x$ near~$0$. Thus, $C^P ((0,x,-x),{\bm  t})$ is negative
when  $t_2$ is close to~$t_1$ and $x$ is small enough, which gives a PMS.

\begin{corollary}\label{cor_price_mod_cont}
Let $f\in \cC^2(\R)$. Under Assumption~\ref{assumption_mod_price}, the function
$h_{P,t}$ is $\cC^1(\R)$, bijective on~$\R$ and such that $h_{P,t}'>0$. Thus, the
result of Theorem~\ref{thm_price_mod_cont} holds and the last trade 
$\xi^\star_{T}$ has the same sign as $-\vi$.

 Besides, if~\eqref{cond_ttpms_p_bs} also holds, $\xi^\star_t$  has the
same sign as~$-\vi$ for any $0<t<T$, which rules out TTPM.
\end{corollary}

As for model~$V$, we consider now the case of a power-law shape
$f(x)=|x|^\gamma$. We can apply the results of
Theorem~\ref{thm_price_mod_cont} in this case. We can also notice
from~\eqref{dynamic_modelsVP} that
$dE_t=(\eta_t-\rho_t(1+\gamma))E_tdt$. Therefore, model~$P$ with
resilience $\rho_t$ is the same as model~$V$ with
resilience $\tilde{\rho}_t=\rho_t(1+\gamma)-\eta_t$.
\begin{corollary}\label{cor_powerlaw_p}
When $f(x)=|x|^\gamma$,  model~$P$ does not
admit PMS if, and only if 
 \begin{equation*}
\forall t\ge 0, \  \rho_t(2+\gamma)-\eta_t \ge 0.
\end{equation*}
It does not
admit transaction-triggered price manipulation  if and only if 
\begin{eqnarray*}
&&\forall t\ge 0, \  \rho_t(1+\gamma)-\eta_t \ge 0,\ \text{and }
\left(\frac{\rho_t(1+\gamma)-\eta_t }{\rho_t(2+\gamma)-\eta_t} \right)'+\rho_t\left(\frac{\rho_t(1+\gamma)-\eta_t}{\rho_t(2+\gamma)+\eta_t }\right)\ge 0.
\end{eqnarray*}
\end{corollary}

\section{Proofs}
\label{section_proofs}

\subsection{The block shape  case}

\begin{proofT}{thm:m_v_cte}
The quadratic form~\eqref{eq:cost_model_vol_const_deterministe} is given by
$C^V ({\bm \xi}, {\bm t}) =\frac{1}{2}{\bm \xi}^T M^V {\bm \xi}$, with
$M^V_{i,j}=\frac{\exp\left(-\left|\int_{t_i}^{t_j} \rho_s ds
    \right|\right)}{\lambda(t_i \vee t_j)}$, $0\le i,j \le N$.
Let us assume that $a_i\tilde{a}_i<1,\forall i\in\left\{1,\ldots ,N\right\}$. Then, we can define the following  vectors:
\begin{equation*}
\mathbf{y}_{0}  =  \frac{\mathbf{e}_{0}}{\sqrt{\lambda(t_{0})}},\
\mathbf{y}_{i} =
\tilde{a}_{i}\mathbf{y}_{i-1}+\frac{\mathbf{e}_{i}}{\sqrt{\lambda(t_{i})}}\sqrt{1-a_{i}\tilde{a}_i}
, \ 1 \le i\le N
\end{equation*}
where $\mathbf{e}_{0}\ldots \mathbf{e}_{N}$ denote the canonical basis of
$\mathbb{R}^{N+1}$. We have $M^{V}_{ij}=\mathbf{y}_{i}^T \mathbf{y}_{j}$. We introduce $Y$ the upper triangular matrix with columns
$\mathbf{y}_{0},\ldots, \mathbf{y}_{N}$. By assumption, it is invertible and
so is $M=Y^TY$. Conversely, if $M^V$ is positive definite, the minors
$$\det((M^V_{i,j})_{0\le i,j\le
  n})=\frac{1}{\lambda(t_0)}\prod^n_{i=1}\frac{1}{\lambda(t_i)}(1-a_i\tilde{a}_i), \ 1\le n\le N$$
are positive, which gives~\eqref{eq:condit_uniq_sol_vol_mod}.

Let us turn to the optimization problem. One has to minimize $C^V ({\bm \xi},
{\bm t})$ under the linear constraint $\sum_{i=0}^N\xi_i=-\vi$, which gives
\begin{equation}
  \label{eq:opt_strat_model_vol_cte}
  {\bm \xi}^{\star}=-\frac{\vi}{\mathbf{1}^T\left(M^{V}\right)^{-1}\mathbf{1}}\left(M^{V}\right)^{-1}\mathbf{1},
\end{equation}
where $\mathbf{1}\in \R^{N+1}$ is a vector of ones. Since $Y$ is upper
triangular, it can be easily inverted and we can calculate explicitly
$\left(M^{V}\right)^{-1}\mathbf{1}$ and get~\eqref{opt_sol_vol_mod_const}.
\end{proofT}

\subsection{General LOB shape with model~$V$}\label{proof_model_vol_general_case}

%% In this part, we prove the results given in Section
%% \ref{results_general_case}, where we consider that the distribution of orders
%% in the limit order book can be represented by the product of two functions
%% $\lambda(t)f(\xi)$ where $f:\mathbb{R}\mapsto[0,\infty)$ and
%% $\lambda:[0,T]\mapsto(0,\infty)$ are continuous functions.

%% Unlike the case of a constant shape function for the limit order book, there is no relationship between the two models and therefore,  we describe in detail each of them.

Let us introduce some notations. For the time grid~${\bm t}$ given
by~\eqref{det_time_grid}, we introduce the next quantities:
\begin{equation}
\label{alpha definition}
\alpha_{k}:=\int^{t_{k}}_{t_{k-1}}\rho_{s}ds,\: k=1,\ldots,N.
\end{equation}
We can write the cost function~\eqref{CM_def} as follows
\begin{equation}\label{eq_to_min_vol_mod_gral}
  C^V ({\bm \xi}, {\bm
    t})=\sum^{N}_{n=0}\lambda(t_{n})\left[G\left(\frac{E_{{n}}+\xi_{n}}{\lambda(t_{n})}\right)-G\left(\frac{E_{{n}}}{\lambda(t_{n})}\right)\right],
\end{equation}
where we use the following  notations (observe that $E_n=a_n(E_{n-1}+\xi_{n-1})$)
$$E_{0}=0,\  E_n=\sum^{n-1}_{i=0}\xi_{i}e^{-\sum^{n}_{k=i+1}\alpha_{k}},\: 1\leq n\leq N.$$

\begin{lemma}
\label{lemma_2_vol}  
We have $\frac{\partial C^{V}}{\partial \xi_{N}}=F^{-1}\left(\frac{E_{N}+\xi_{N}}{\lambda(t_{N})}\right)$ and, for $i=0,\ldots ,N-1$,
\begin{equation}
\label{relation derivatives model 1}
\frac{\partial C^{V}}{\partial \xi_{i}}-a_{i+1}\frac{\partial C^{V}}{\partial \xi_{i+1}}=F^{-1}\left(\frac{E_{i}+\xi_{i}}{\lambda(t_{i})}\right)-a_{i+1}F^{-1}\left(\frac{E_{i+1}}{\lambda(t_{i+1})}\right).
\end{equation}  
\end{lemma}  

\begin{proof}%{lemma_2_vol}
Let us  first observe that $
\frac{\partial E_{n}}{\partial \xi_{i}}=0,\:\text{ if } i\geq n,\:\text {and }
\frac{\partial E_{n}}{\partial \xi_{i}}=e^{-\sum^{n}_{k=i+1}\alpha_{k}}\:\text{ if
} i<n$.
Thus, we get by using that $G'=F^{-1}$:
{
\begin{eqnarray*}
\frac{\partial C^{V}}{\partial \xi_{i}}%%  & = &
%% G'\left(\frac{E_{i}+\xi_{i}}{\lambda(t_{i})}\right) +  \sum^{N}_{n=i+1}\lambda(t_{n})\left[G'\left(\frac{E_{n}+\xi_{n}}{\lambda(t_{n})}\right)\frac{e^{-\sum^{n}_{k=i+1}\alpha_{k}}}{\lambda(t_{n})}-G'\left(\frac{E_{n}}{\lambda(t_{n})}\right)\frac{e^{-\sum^{n}_{k=i+1}\alpha_{k}}}{\lambda(t_{n})}\right]\\
& = &
F^{-1}\left(\frac{E_{i}+\xi_{i}}{\lambda(t_{i})}\right)+\sum^{N}_{n=i+1}e^{-\sum^{n}_{k=i+1}\alpha_{k}}\left(F^{-1}\left(\frac{E_{n}+\xi_{n}}{\lambda(t_{n})}\right)-F^{-1}\left(\frac{E_{n}}{\lambda(t_{n})}\right)\right)\\
& = &
F^{-1}\left(\frac{E_{i}+\xi_{i}}{\lambda(t_{i})}\right)-e^{-\alpha_{i+1}}F^{-1}\left(\frac{E_{i+1}}{\lambda(t_{i+1})}\right)\\
& + &
e^{-\alpha_{i+1}}\left[F^{-1}\left(\frac{E_{i+1}+\xi_{i+1}}{\lambda(t_{i+1})}\right)+\sum^{N}_{n=i+2}e^{-\sum^{n}_{k=i+2}\alpha_k}\left(F^{-1}\left(\frac{E_{n}+\xi_{n}}{\lambda(t_{n})}\right)-F^{-1}\left(\frac{E_{n}}{\lambda(t_{n})}\right)\right)\right]
\\
& = & F^{-1}\left(\frac{E_{i}+\xi_{i}}{\lambda(t_{i})}\right)-a_{i+1}F^{-1}\left(\frac{E_{i+1}}{\lambda(t_{i+1})}\right)+a_{i+1}\frac{\partial C^{V}}{\partial \xi_{i+1}}.
\end{eqnarray*}  
}
\end{proof}

\begin{lemma}
\label{lemma_1_vol}
Under Assumption~\ref{assumption_mod_vol}, we obtain the next conclusions.
\begin{enumerate}
\item For $i\in \left\{1,\ldots
    ,N\right\}$, the function $ h_{V,i}$ defined in~\eqref{def_hvi}
  is an increasing bijection on $\mathbb{R}$ that satisfies $\sgn(x)
  h_{V,i}(x)\ge \frac{1-a_i\tilde{a}_{i}}{1-a_i}F^{-1}(x)$. 

\item If~\eqref{eq:condit_manip_price_vol_mod} holds, then we
    have 
$\sgn(x) h^{-1}_{V,i+1}\left(x\right)\geq \sgn(x)  \tilde{a}_{i}h^{-1}_{V,i}\left(x\right)$
for $i\in\left\{1,\ldots ,N-1\right\}$.

\item $\sgn(x) F(x) \ge \sgn(x) \tilde{a}_{N}h^{-1}_{V,N}\left(x \right)$.
\end{enumerate}  
\end{lemma}

\begin{proof}
\begin{enumerate}
\item {Since the resilience $\rho_t$ is positive, we have $0<a_i<1$, and
    $\tilde{a}_i\le 1$ since $\rho_t+\eta_t\ge 0$ by
    Assumption~\ref{assumption_mod_vol}. We then get
\begin{equation*}
\frac{\partial h_{V,i}(x)}{\partial
  x}=\frac{1}{1-a_i}\left[\frac{1}{f(F^{-1}(x))}-\frac{a_{i}\tilde{a}_{i}}{f(F^{-1}(\tilde{a}_{i}x))}\right]\ge
\frac{1-a_{i}\tilde{a}_{i}}{1-a_i}\frac{1}{f(F^{-1}(x))}>0
\end{equation*}  
because $f$ is nondecreasing on $\mathbb{R_{-}}$ and nonincreasing on
$\mathbb{R_{+}}$,  and $F^{-1}$ is increasing.}

\item{We set $\hat{f}(x)=(F^{-1})'(x)=1/f(F^{-1}(x))$: this function is positive,
    nonincreasing on $\R_-$ and nondecreasing on $\R_+$. Let $\nu\ge 0$ and $y = h^{-1}_{V,i+1}(\nu)$. We note that
$y \ge 0$ because $h_{V,i+1}(0)=0$ and $h_{V,i+1}$ is increasing by
the first point of this lemma. Thus, we have that
\begin{eqnarray*}
\nu & = &
\frac{F^{-1}(y)-a_{i+1}F^{-1}(\tilde{a}_{i+1}y)}{1-a_{i+1}}\\
& = &
F^{-1}(\tilde{a}_{i+1}y)+\frac{F^{-1}(y)-F^{-1}(\tilde{a}_{i+1}y)}{1-a_{i+1}}\\
& = &
F^{-1}(\tilde{a}_{i+1}y)+\frac{1}{1-a_{i+1}}\int^{y}_{\tilde{a}_{i+1}y}\hat{f}(\xi)d\xi \le  F^{-1}(y)+\frac{1-\tilde{a}_{i+1}}{1-a_{i+1}}y\hat{f}(y)=:g_{i+1}(y)
\end{eqnarray*}
Hence, we obtain that $g_{i+1}$ is increasing on $\mathbb{R}$ and then, $y\ge g^{-1}_{i+1}(\nu)$.
Let $z = \tilde{a}_{i}h^{-1}_{V,i}\left(\nu\right)\ge0$. We have:
\begin{eqnarray*}
\nu & = &
\frac{F^{-1}\left(\frac{z}{\tilde{a}_{i}}\right)-a_{i}F^{-1}(z)}{1-a_{i}}\\
& = & F^{-1}(z)+\frac{F^{-1}\left(\frac{z}{\tilde{a}_{i}}\right)-F^{-1}(z)}{1-a_{i}}\\
& = &
F^{-1}(z)+\frac{1}{1-a_{i}}\int^{\frac{z}{\tilde{a}_{i}}}_{z}\hat{f}(\xi)d\xi
 \ge
F^{-1}(z)+\frac{\left(\frac{1}{\tilde{a}_{i}}-1\right)}{1-a_{i}}z\hat{f}(z)=:\bar{g}_i(z)
\end{eqnarray*}
Therefore, if~\eqref{eq:condit_manip_price_vol_mod} holds, we get that
$g_{i+1}(x)\leq\bar{g}_i(x)$ for all $x\ge0$. Then, we have
 $g^{-1}_{i+1}(x)\geq g^{-1}_i(x)$, and therefore
$$y\ge g^{-1}_{i+1}(\nu)\geq g^{-1}_i(\nu) \ge z.$$
The same arguments for $\nu\le 0$ give $y\le g^{-1}_{i+1}(\nu)\leq g^{-1}_i(\nu) \le z$.
}
\item Using the above definition, we have $\sgn(x) \bar{g}_{N}(x)\geq
    \sgn(x) F^{-1}(x)$, and therefore we get
\begin{equation*}
\sgn(\nu) F(\nu)\geq \sgn(\nu)\bar{g}^{-1}_{N}(\nu)\ge \sgn(\nu) z=\sgn(\nu)\tilde{a}_Nh^{-1}_{V,N}\left(\nu\right).
\end{equation*}
\end{enumerate}
\end{proof}

\begin{lemma}\label{lemma_3_vol}
Let $a\in(0,1)$ and $b>0$ such that $ab\le 1$. We have
$G(x)-\frac{1}{b}G(abx)\ge 0$ for $x\in \R$, and  $G(x)-\frac{1}{b}G(abx)\underset{|x|\rightarrow + \infty }\rightarrow + \infty$.
\end{lemma}
\begin{proof} Since $G$ is convex ($G'=F^{-1}$ is increasing) and $G(0)=0$, $G(abx)\le
  abG(x)$. If $b > 1$, we then have $G(x)-\frac{1}{b}G(abx)\ge G(x)(1-a)$
  which gives the result. If $b\le 1$, we have
\begin{eqnarray*}
G(x)-\frac{1}{b}G(ab x) & = & \int^{x}_{0}F^{-1}(u)du-\frac{1}{b}\int^{ab x}_{0}F^{-1}(u)du 
 =  \int^{x}_{0}F^{-1}(u)du-\int^{ax}_{0}F^{-1}(b v)dv \\
& = & \int^{x}_{ax}F^{-1}(u)du+\int^{ax}_{0}\left(F^{-1}(u)-F^{-1}(b
  u)\right)du  \ge  |x|(1-a)F^{-1}(|ax|) \underset{|x|\rightarrow + \infty
}\rightarrow \infty. \qedhere
\end{eqnarray*}
\end{proof}

\begin{proofT}{thm_vol_mod} We rewrite the cost function~\eqref{eq_to_min_vol_mod_gral}
  to minimize as follows:
\begin{eqnarray*}
C^V ({\bm \xi}, {\bm
    t}) & = &
\sum^{N}_{n=0}\lambda(t_{n})\left[G\left(\frac{E_{n}+\xi_{n}}{\lambda(t_{n})}\right)-G\left(\frac{E_{n}}{\lambda(t_{n})}\right)\right]
\\
& = &
\lambda(t_{N})G\left(\frac{\sum^{N}_{i=0}\xi_{i}e^{-\sum^{N}_{k=i+1}\alpha_{k}}}{\lambda(t_{N})}\right)-\lambda(0)G(0)\\
& + & \sum^{N-1}_{n=0}\left[\lambda(t_{n})G\left(\frac{\sum^{n}_{i=0}\xi_{i}e^{-\sum^{n}_{k=i+1}\alpha_{k}}}{\lambda(t_{n})}\right)-\lambda(t_{n+1})G\left(\frac{e^{-\alpha_{n+1}}\sum^{n}_{i=0}\xi_{i}e^{-\sum^{n}_{k=i+1}\alpha_{k}}}{\lambda(t_{n+1})}\right)\right]
\end{eqnarray*}
We define the linear map $T:\mathbb{R}^{N+1}\rightarrow\mathbb{R}^{N+1}$ by $
(T\xi)_{n}=\frac{\sum^{n}_{i=0}\xi_{i}e^{-\sum^{n}_{k=i+1}\alpha_{k}}}{\lambda(t_{n})}$,
so that
\begin{equation}
\label{bounded_cost_vol}  
C^V ({\bm \xi}, {\bm
    t}) =\lambda(t_{N})G((T\xi)_{N}) + \sum^{N-1}_{n=0}\left[\lambda(t_{n})G((T\xi)_{n})-\lambda(t_{n+1})G\left(\tilde{a}_{n+1}(T\xi)_{n}\right)\right].
\end{equation}  
Let us observe that $T$ is a linear bijection. By Lemma~\ref{lemma_3_vol} we
get that $C^V ({\bm \xi}, {\bm t})\ge 0$ and $C^V ({\bm \xi}, {\bm
  t})\underset{| {\bm \xi}|\rightarrow +\infty}\rightarrow +\infty$,
which gives the existence of a minimizer~${\bm \xi}^\star$ over ${\bm \xi}$, s.t. $\sum_{i=0}^N
\xi_i=-\vi$.
Thus, by  using~\eqref{relation derivatives model 1}, there must be a  Lagrange multiplier~$\nu$ such that
\begin{equation}
\label{relation between nu and h1}  
\nu=h_{V,i+1}\left(\frac{E_{i}+\xi^\star_{i}}{\lambda(t_{i})}\right),\:
i=0\ldots N-1, \text { and }\nu=F^{-1}\left(\frac{E_{N}+\xi^\star_{N}}{\lambda(t_{N})}\right).
\end{equation}  
We have $\frac{E_{i}+\xi^\star_{i}}{\lambda(t_{i})}=h^{-1}_{V,i+1}\left(\nu\right)$ and
then $E_{i+1}=\lambda(t_{i})a_{i+1}h^{-1}_{V,i+1}\left(\nu\right)$, for $0\le i\le N-1$.
Thus, we get
\begin{eqnarray*}
\xi^\star_{0}&=&\lambda(t_{0})h^{-1}_{V,1}\left(\nu\right),  \\
\xi^\star_{i}&=&\lambda(t_{i})h^{-1}_{V,i+1}\left(\nu\right)-\lambda(t_{i-1})a_{i}h^{-1}_{V,i}\left(\nu\right),\
1\le i\le N-1, \\
\xi^\star_{N}&=&F(\nu)\lambda(t_{N})-\lambda(t_{N-1})a_{N}h^{-1}_{V,N}\left(\nu\right)
\end{eqnarray*}
Furthermore, we note that 
\begin{equation*}
\sum^N_{i=0}\xi^\star_i=-\vi=\lambda(t_0)(1-a_1)h^{-1}_{V,1}(\nu)+\ldots +\lambda(t_{N-1})(1-a_N)h^{-1}_{V,N}(\nu)+F(\nu)\lambda(t_N).  
\end{equation*}
By Lemma~\ref{lemma_1_vol} The right side is an increasing bijection on~$\R$, and we deduce that there is
only one $\nu \in \R$ which satisfies the above equation. This give the
uniqueness of the minimizer~${\bm \xi}^\star$. Moreover, the
functions $F^{-1}$ and $h_{V,i}$ vanish in~$0$, and $\nu$ has the same sign as
$-\vi$, which gives that $\xi^\star_{0}$ and $\xi^\star_{N}$ have the same
sign as $-\vi$ by Lemma~\ref{lemma_1_vol}. Besides, if~\eqref{eq:condit_manip_price_vol_mod} holds, the
trades $\xi^\star_{i}$ have also the same sign as~$-\vi$.
\end{proofT}

Let us now prepare the proof of Theorem~\ref{thm_vol_mod_cont} and assume
that~$h_{V,t}$ is bijective increasing. We introduce for $0\le t\le T$,
\begin{equation}\label{def_cout}
  C^V(t,T,E_t,X_t)=\lambda(t)\left[G(\zeta_t)
  -G\left(\frac{E_t}{\lambda(t)}\right)\right]+\int_t^T F^{-1}(\zeta_u) \xi_u du
  +\lambda(T)[G(F(\nu))-G(\zeta_T)],
\end{equation}
where
\begin{eqnarray}
 &&\nu\in \R, s.t. -E_t+\int_t^T \lambda(u) \rho_u h_{V,u}^{-1}(\nu)du
 +\lambda(T)F(\nu)=-X_t, \label{def_nu_cost}\\
&&\zeta_u=h_{V,u}^{-1}(\nu), \ \xi_u=\lambda(u)[\frac{d\zeta_u}{du}+(\rho_u+\eta_u)\zeta_u].
\end{eqnarray}
Let us observe that $\nu \mapsto\int_t^T \lambda(u) \rho_u h_{V,u}^{-1}(\nu)du
+\lambda(T)F(\nu)$ is increasing an bijective on~$\R$, and~\eqref{def_nu_cost} admits a
unique solution.
The function $C^V(t,T,E_t,X_t)$ denotes the minimal cost to liquidate~$X_t$
shares on the time interval~$[t,T]$ given the current state~$E_t$. In
particular, we observe that
$$C^V(T,T,E_T,X_T)=\lambda(T)\left[G\left(\frac{E_T-X_T}{\lambda(T)} \right)- G\left(
    \frac{E_T}{\lambda(T)} \right)\right], $$
which is the cost of selling~$X_T$ shares at
time~$T$. Besides, an integration by parts gives that
\begin{equation}\label{cout_IPP}
  C^V(t,T,E_t,X_t)=-\lambda(t)G\left(\frac{E_t}{\lambda(t)}\right)
  +\int_t^T  \lambda(u)\left[
    (\rho_u+\eta_u) F^{-1}(\zeta_u)\zeta_u-\eta_uG(\zeta_u) \right]du+ \lambda(T)G(F(\nu)).
\end{equation}
The function $\zeta \mapsto (\rho_u+\eta_u)
F^{-1}(\zeta)\zeta-\eta_uG(\zeta)$ is nonnegative since it vanishes at~$0$,
and its derivative is equal to~$\rho_u h_{V,u}(\zeta)$ that has the same sign
as~$\zeta$. Since $G\ge 0$, we get:
\begin{equation}\label{cout_positif}
   C^V(0,T,0,\vi)\ge 0.
\end{equation}
 
Formula~\eqref{def_cout} can be guessed by simple but tedious
calculations: one has to consider the associated discrete problem on a regular
time-grid and then let the time-step going to zero. We do not present these
calculations here since we will prove directly by a verification argument that
this is indeed the minimal cost.

\begin{proofT}{thm_vol_mod_cont}
Let $(X_t, 0\le t\le T+)$ denote an admissible strategy that
liquidates~$\vi$. We consider $(E_t,0\le t\le T+)$ the solution of
$dE_t=dX_t-\rho_tE_tdt$, $\nu_t$ the solution of~\eqref{def_nu_cost} and
$\zeta_t=h_{V,t}^{-1}(\nu_t)$. We set
$$C_t=\int_0^t  
  F^{-1}\left(\frac{E_s}{ \lambda(s)} \right) dX^c_s +\sum_{0\le s< t}
  \lambda(s)\left[G\left(\frac{E_s+\Delta X_s}{ \lambda(s)}
    \right)-G\left(\frac{E_s}{ \lambda(s)} \right)\right] +C^V(t,T,E_t,X_t). $$
  Let us observe that $C_T=C^V(X)$ and $C_0=C^V(0,T,0,\vi)$. We are going to
  show that $dC_t\ge 0$, and that $dC_t=0$ holds only for~$X^\star$. This will
  in particular show that $C^V(X)\ge 0$ from~\eqref{cout_positif}.

Let us first consider the case of a jump $\Delta X_t>0$. Then, we have
$$\Delta C_t=\lambda(t)\left[G\left(\frac{E_t+\Delta X_t}{ \lambda(t)}
    \right)-G\left(\frac{E_t}{ \lambda(t)} \right)\right]+C^V(t+,T,E_{t+},X_{t+})-C^V(t,T,E_t,X_t).$$
Since $\Delta E_t =\Delta X_t$, the solution~$\nu_t$ of~\eqref{def_nu_cost} is
also the solution of $-E_{t+}+\int_t^T \lambda(u) \rho_u h_{V,u}^{-1}(\nu_t)du
 +\lambda(T)F(\nu_t)=-X_{t+}$, and then $\Delta C_t=0$.  Now, let us calculate $dC_t$. We set
 \begin{eqnarray*}
    \tilde{C}(t,T,E_t,X_t,v)
%%   &=&\lambda(t)\left[G(h_{V,t}^{-1}(v))
%%   -G\left(\frac{E_t}{\lambda(t)}\right)\right]+\lambda(T)[G(F(v))-G(h_{V,T}^{-1}(v))]\\&&
%%   +\int_t^T  \lambda(u) F^{-1}(h_{V,u}^{-1}(v))
%%   \left[\frac{dh_{V,u}^{-1}(v)}{du}+(\rho_u+\eta_u)h_{V,u}^{-1}(v) \right] du
%%   \\
   &=&\lambda(T)G(F(v)) -\lambda(t)G\left(\frac{E_t}{\lambda(t)}\right)\\
  &&+\int_t^T  \lambda(u)\left[ (\rho_u+\eta_u) F^{-1}(h_{V,u}^{-1}(v))
  h_{V,u}^{-1}(v)   - \eta_u G(h_{V,u}^{-1}(v)) \right]du .
\end{eqnarray*}
Then, we have from~\eqref{cout_IPP}:
\begin{eqnarray*}dC_t&=&F^{-1}\left(\frac{E_t}{ \lambda(t)} \right) dX^c_t-
  \lambda'(t)G\left(\frac{E_t}{\lambda(t)}\right) dt-F^{-1}\left(\frac{E_t}{ \lambda(t)}
\right)(dX^c_t-(\rho_t+\eta_t)E_tdt)\\
&&-\lambda(t)(\rho_t+\eta_t)F^{-1}(\zeta_t)\zeta_t dt +\lambda'(t)G(\zeta_t)dt
+\frac{\partial
  \tilde{C}}{\partial v}(t,T,E_t,X_t,\nu_t) d\nu_t.
\end{eqnarray*}
Since $\left[\lambda(T)  f(\nu_t)+\int_t^T
\lambda(u) \rho_u (h_{V,u}^{-1})'(\nu_t) du\right]d\nu_t-\lambda(t)\rho_th_{V,t}^{-1}(\nu_t)dt =d(E_t-X_t)=-\rho_t E_tdt$ and
\begin{eqnarray*}
  \partial_v\tilde{C}(t,T,E_t,X_t,v)&=&\lambda(T) v f(v)+\int_t^T
\lambda(u) \rho_u (h_{V,u}^{-1})'(v)
\left[F^{-1}(h_{V,u}^{-1}(v))+\frac{\rho_u+\eta_u}{\rho_u}
\frac{h_{V,u}^{-1}(v)}{f(h_{V,u}^{-1}(v))} \right]du \\
&=&v\left[\lambda(T)  f(v)+\int_t^T
\lambda(u) \rho_u (h_{V,u}^{-1})'(v)
du \right],
\end{eqnarray*}
we finally get
\begin{eqnarray}
dC_t&=&\lambda(t)\left[(\rho_t+\eta_t)\left(\frac{E_t}{\lambda(t)}F^{-1}\left(\frac{E_t}{\lambda(t)}\right)-\zeta_tF^{-1}(\zeta_t)\right)+\eta_t\left(G(\zeta_t)-G\left(\frac{E_t}{\lambda(t)}\right)\right)+\rho_t
  h_{V,t}(\zeta_t)\left(\zeta_t-\frac{E_t}{\lambda(t)} \right) \right]dt \nonumber\\
&:=&\lambda(t)\psi_t(\zeta_t)dt.
\end{eqnarray}
We have
$\psi_t'(\zeta)=-(\rho_t+\eta_t)\left(F^{-1}(\zeta)+\frac{\zeta}{f(F^{-1}(\zeta))}\right)+\eta_tF^{-1}(\zeta)+\rho_t
h_{V,t}(\zeta)+\rho_t h_{V,t}'(\zeta)(\zeta-\frac{E_t}{\lambda(t)} )=\rho_t
h_{V,t}'(\zeta)(\zeta-\frac{E_t}{\lambda(t)} )$. Since $h_{V,t}'>0$, $\psi_t$
vanishes at $\zeta=\frac{E_t}{\lambda(t)}$, and is positive for
$\zeta\not = \frac{E_t}{\lambda(t)}$.

Thus, if $X$ is an optimal strategy, we
necessarily have $\zeta_t=\frac{E_t}{\lambda(t)}$, $dt$-a.e. Then, we get by differentiating 
$\left[X_t-E_t+\int_t^T \lambda(u)\rho_u h_{V,u}^{-1}(\nu_t)
  du+\lambda(T)F(\nu_t)\right]=0$ that $\left[\int_t^T \lambda(u)\rho_u (h_{V,u}^{-1})'(\nu_t) du +\lambda(T)f(\nu_t)\right]d\nu_t =0,$
which gives $d\nu_t=0$ since $(h_{V,u}^{-1})'>0$ and $f>0$. Thus, we get that $\nu_t=\nu$ where $\nu$
is the solution of~\eqref{def_nu_vi_vol}. In particular, we get
$\Delta X_0=E_{0+}=\lambda(0)h_{V,0}^{-1}(0)=\Delta X^\star_0$ and
then~$X=X^\star$, which gives the uniqueness of the optimal strategy. Last we
observe that $\nu$ has the same sign as $-\vi$ and thus  $\xi^\star_0$ has the same sign as $-\vi$.
\end{proofT}

\begin{proofC}{cor_vol_mod_cont}
Since $\rho_t+\eta_t\ge 0$ and $xf'(F^{-1}(x))\ge 0$ by
Assumption~\ref{assumption_mod_vol}, we have $$h'_{V,t}(x)=\frac{\eta_t+2\rho_t}{\rho_t} \frac{1}{f(F^{-1}(x))} -
\frac{\eta_t+\rho_t}{\rho_t} \frac{xf'(F^{-1}(x))}{f(F^{-1}(x))^3} >0.$$
Also, we have $\sgn(x)h_{V,t}(x)\ge\sgn(x) F^{-1}(x)$ and then $\sgn(x)
h_{V,t}^{-1}(x)\le \sgn(x) F(x)$, which gives that the
last trade~$\xi^\star_T$ has the same sign as~$-\vi$.
%% It remains to look at the nonnegativity of the trades. Let us first observe
%% that $\nu$ has the same sign as $-\vi$, since $h_{V,t}$ in increasing and
%% vanishes at~$0$. Then, $\xi^\star_0$ has the same sign as $-\vi$. Besides, we
%% have $\sgn(x) h_{V,t}(x)\ge \sgn(x) F^{-1}(x)$ and then $\sgn(x)
%% h_{V,t}^{-1}(x)\le \sgn(x) F(x)$, which gives that $\xi^\star_T$ has the same
%% sign as $-\vi$.
Then, we have
$\frac{d\zeta_t}{dt}=-\frac{1}{h_{V,t}'(\zeta_t)}\frac{dh_{V,t}}{dt}(\zeta_t)$
and thus
\begin{eqnarray*}
  \xi^\star_t&=&\frac{\lambda(t)\zeta_t }{h_{V,t}'(\zeta_t)} \left[
  -\frac{d\left(\eta_t/\rho_t \right)}{dt}
  \frac{1}{f(F^{-1}(\zeta_t))}+(\rho_t+\eta_t)h_{V,t}'(\zeta_t)\right] \\
&=&\frac{\lambda(t)\zeta_t }{h_{V,t}'(\zeta_t)}
\left[\frac{1}{\rho_tf(F^{-1}(\zeta_t))}\left(\frac{\rho_t'\eta_t-\rho_t\eta_t'}{\rho_t}+(\rho_t+\eta_t)(2
      \rho_t+\eta_t)\right)-\frac{(\eta_t+\rho_t)^2}{\rho_t}\frac{\zeta_t
    f'(\zeta_t)}{f(F^{-1}(\zeta_t))^3} \right]
\end{eqnarray*}
is nonnegative if~\eqref{cond_ttpms_v_bs} holds since $h_{V,t}'>0$ and $\zeta_t f'(\zeta_t)\ge 0$. 
\end{proofC}

\begin{lemma}\label{lemma_implication} We have $\eqref{eq:condit_manip_price_vol_mod} \implies
  \eqref{eq:condit_t_t_manip_m_v}$ if $\rho_t+\eta_t \ge 0,\ t\ge 0$.
\end{lemma}
\begin{proof}
We have
\begin{eqnarray*}\eqref{eq:condit_manip_price_vol_mod} \Leftrightarrow
\frac{1}{\tilde{a}_i}\frac{1-\tilde{a}_i}{1-a_i}\geq\frac{1-\tilde{a}_{i+1}}{1-a_{i+1}} 
&\Leftrightarrow&
\left(1-a_{i+1}\right)-\tilde{a}_i\left(1-a_{i+1}\right)\geq\tilde{a}_i\left(1-a_i\right)-\tilde{a}_i\tilde{a}_{i+1}\left(1-a_i\right) \\
&\Leftrightarrow&
\tilde{a}_{i+1}\left(1-a_i\right)+\frac{1}{\tilde{a}_i}(1-a_{i+1}) \ge
1-a_i+1-a_{i+1}.
\end{eqnarray*}
%% \\
%% & &\Leftrightarrow\tilde{a}_{i+1}\left(1-a_i\right)+\frac{1}{\tilde{a}_i}\left(1-a_{i+1}\right)\ge\lambda\left(1-a_i+1-a_{i+1}\right)
%% \end{eqnarray*}  
Since $\tilde{a}_{i+1} \le 1$,  we get
$
1-a_i+1-a_{i+1}  = 1-a_ia_{i+1}+(1-a_i)(1-a_{i+1}) 
\ge  1-a_ia_{i+1}+\tilde{a}_{i+1}(1-a_i)(1-a_{i+1}) $.  
Thus, $\eqref{eq:condit_manip_price_vol_mod}$ implies that:
\begin{eqnarray*}
\tilde{a}_{i+1}\left(1-a_i\right)+\frac{1}{\tilde{a}_i}(1-a_{i+1})\geq
1-a_ia_{i+1}+\tilde{a}_{i+1}(1-a_i)(1-a_{i+1}) \\
%% \Leftrightarrow \frac{1}{\tilde{a}_i}(1-a_{i+1})\geq
%% 1-a_ia_{i+1}-\tilde{a}_{i+1}a_{i+1}(1-a_i) \\
\Leftrightarrow 1-\tilde{a}_i+a_ia_{i+1}\tilde{a}_i-a_{i}\tilde{a}_{i+1}\geq
a_{i+1}-\tilde{a}_i\tilde{a}_{i+1}a_{i+1}+a_i\tilde{a}_ia_{i+1}\tilde{a}_{i+1}-\tilde{a}_{i+1}a_{i+1}
\\
\Leftrightarrow\left(1-\tilde{a}_i\right)\left(1-a_{i+1}\tilde{a}_{i+1}\right)\geq
a_{i+1}\left(1-\tilde{a}_{i+1}\right)\left(1-a_i\tilde{a}_i\right) 
\Leftrightarrow \eqref{eq:condit_t_t_manip_m_v}. \qedhere
\end{eqnarray*}
\end{proof}

\subsection{General LOB shape with model~$P$}
We first focus on discrete strategies on the time grid~${\bm t}$ such
as~\eqref{det_time_grid}. We introduce the following shorthand notation
$D_n=D_{t_n}$ for $0\le n\le N$ and have
$$ D_0=0, \ D_{n}=a_nF^{-1}\left(\frac{\xi_{n-1}}{\lambda(t_{n-1})}+F(D_{n-1})\right), \ 1\le n\le N.$$
We can write the cost function~\eqref{CM_def} as follows:
\begin{equation}  \label{eq_to_min_price_mod_gral}
C^P ({\bm \xi}, {\bm
    t})  = 
\sum^{N}_{n=0}\lambda(t_{n})\int^{D_{t_{n}+}}_{D_{t_{n}}}x f(x)dx= \sum^N_{n=0}\lambda(t_n)\left[ G\left(
    \frac{\lambda(t_n)F(D_{n})+\xi_n}{\lambda(t_n)} \right)-G(F(D_{n}))
\right].
\end{equation}

%% To prove the results given in Section~\ref{results_general_case_price}, we use the same ideas that in the
%% case of the volume impact reversion model in Section
%% \ref{proof_model_vol_general_case}.

We begin with the following lemmas that we use to characterize the
critical points of the optimization problem.

\begin{lemma}
\label{lemma_1_price_model_general_case}  
For $i=0,\ldots ,N-1$, we have the following equations:
\begin{equation*}
\frac{\partial C^P}{\partial
  \xi_i}=F^{-1}\left(\frac{\xi_i}{\lambda(t_i)}+F(D_i)\right)+\hat{a}_{i+1}\frac{f(D_{i+1})}{f\left(F^{-1}\left(\frac{\xi_{i}}{\lambda(t_i)}+F(D_i)\right)\right)}\left(\frac{\partial
  C^P}{\partial \xi_{i+1}}-D_{i+1}\right).
\end{equation*}  
\end{lemma}  

\begin{proof}
  First, we have $\frac{\partial D_n}{\partial \xi_i}=0$ for $i\geq n$, and the following recursive equations:
 \begin{eqnarray*}
&&  \frac{\partial D_n}{\partial \xi_{n-1}}=\frac{a_n}{\lambda(t_{n-1})f
  \left(F^{-1}(\frac{\xi_{n-1}}{\lambda(t_{n-1})}+F(D_{n-1})) \right)},   \ \frac{\partial D_n}{\partial
   \xi_i}=\frac{\hat{a}_{i+1}f(D_{i+1})}{f\left(F^{-1}(\frac{\xi_i}{\lambda(t_i)}+F(D_i))\right)}\frac{\partial
 D_n}{\partial\xi_{i+1}}\text{ for }1\leq i\leq n-2.  
\end{eqnarray*}

From~\eqref{eq_to_min_price_mod_gral}, we get:
\begin{eqnarray*}
  \frac{\partial
    C^P}{\partial\xi_i}& = &
  F^{-1}\left(\frac{\xi_i}{\lambda(t_i)}+F(D_i)\right)+\sum^N_{n=i+1}\left[
  F^{-1}\left(F(D_n)+\frac{\xi_n}{\lambda(t_n)}\right)-D_n\right]f(D_n)\frac{\partial D_n}{\partial \xi_i}
\\
&=&  F^{-1}\left(\frac{\xi_i}{\lambda(t_i)}+F(D_i)\right)+
\frac{\hat{a}_{i+1}f(D_{i+1})}{f\left(F^{-1}(\frac{\xi_i}{\lambda(t_i)}+F(D_i)
  )\right)}\left[F^{-1}(F(D_{i+1})+\frac{\xi_{i+1}}{\lambda(t_{i+1})})-D_{i+1}\right] \\
&&+\frac{\hat{a}_{i+1}f(D_{i+1})}{f\left(F^{-1}(\frac{\xi_i}{\lambda(t_i)}+F(D_i)
  )\right)}\left[\frac{\partial
    C^P}{\partial\xi_{i+1}}
  -F^{-1}\left(\frac{\xi_{i+1}}{\lambda(t_{i+1})}+F(D_{i+1})\right) \right], %% \\
%% &&
%% \lambda(t_{i+1})f(D_{i+1})\frac{\partial
%%   D_{i+1}}{\partial\xi}\left[F^{-1}\left(\frac{\xi_{i+1}}{\lambda(t_{i+1})}+F(D_{i+1})\right)-D_{i+1}\right]
%% \\
%% & &+\sum^N_{n=i+2}\lambda(t_n)f(D_n)\frac{\partial D_n}{\partial
%%   \xi}\left[F^{-1}\left(\frac{\xi_n}{\lambda(t_n)}+F(D_n)\right)-D_n\right] \\
%% &=&
%% F^{-1}\left(\frac{\xi_i}{\lambda(t_i)}+F(D_i)\right)+\lambda(t_{i+1})f(D_{i+1})a_{i+1}\hat{f}\left(\frac{\xi_i}{\lambda(t_i)}+F(D_i)\right)\frac{1}{\lambda(t_i)}
%% \\
%% & &
%% \left[F^{-1}\left(\frac{\xi_{i+1}}{\lambda(t_{i+1})}+F(D_{i+1})\right)-D_{i+1}\right]+\hat{a}_{i+1}f(D_{i+1})\hat{f}\left(\frac{x_i}{\lambda(t_i)}+F(D_i)\right)
%% \\
%% & & \sum^N_{n=i+2}\lambda(t_n)f(D_n)\frac{\partial D_n}{\partial x_{i+1}}\left[F^{-1}\left(\frac{\xi_n}{\lambda(t_n)}+F(D_n)\right)-D_n\right]
\end{eqnarray*}
which gives the result.
\end{proof}

\begin{lemma}
\label{lemma_2_price_mod}
Under Assumption \ref{assumption_mod_price}, we have that:
\begin{enumerate}
  \item The function $x\mapsto xf(x)$ is increasing on $\mathbb{R}$ (or
equivalently, $\tilde{F}$ is convex).
\item We have $f\left(\frac{x}{a_i}\right)-\hat{a}_if(x) >0, \;\;i=1,\ldots ,N.$  
\item  The function
$$x\in \R,\ h_{P,i}(x)=x\frac{\left[\frac{1}{a_i}f(\frac{x}{a_i})-\hat{a}_if(x)\right]}{f\left(\frac{x}{a_i}\right)-\hat{a}_if(x)}$$
is well-defined, bijective increasing and satisfies $ \sgn(x) h_{P,i}(x) \ge |x|. $
%%The function $x\mapsto\frac{\hat{a}_if(x)}{f\left(\frac{x}{a_i}\right)}$ is
%% nondecreasing on $\mathbb{R}_{+}$ and nonincreasing on $\mathbb{R}_{-}$ and
%% takes values in $(0,1]$
\end{enumerate}    
\end{lemma}

\begin{proof}
1. We have $\left(xf(x)\right)'>0$  since  $x f'(x) \ge 0 $ by Assumption~\ref{assumption_mod_price}.\\
2. We have for $x\in \R $,
$$
\lambda(t_{i-1})f(\frac{x}{a_i})-\lambda(t_i)a_if(x)  \ge \lambda(t_{i-1})f(x)(1-\hat{a}_i) > 0
$$
because $f\left(\frac{x}{a_i}\right)\ge f(x)$ and $\hat{a}_i<1$ by Assumption~\ref{assumption_mod_price}.\\
3. The function $h_{P,i}$ is well-defined thanks to the second point. We have
$ \sgn(x) h_{P,i}(x) \ge |x|$ since
\begin{equation*}
h_{P,i}(x)=x\left[1+\frac{a^{-1}_i}{1-\hat{a}_i\frac{f(x)}{f(\frac{x}{a_i})}} \right],  
\end{equation*}
and it is sufficient to check that $f(x)/f(x/a_i)$ is
 nondecreasing on $\mathbb{R}_{+}$ and nonincreasing on $\mathbb{R}_{-}$.
  We calculate
\begin{equation*}
\left(\frac{f(x)}{f\left(\frac{x}{a_i}\right)}\right)'=\frac{f'(x)f\left(\frac{x}{a_i}\right)-\frac{1}{a_i}f(x)f'\left(\frac{x}{a_i}\right)}{f\left(\frac{x}{a_i}\right)^2}  .
\end{equation*}
This is nonnegative on $\mathbb{R}_{+}$ and nonpositive on $\mathbb{R}_{-}$ if and
only if $\frac{xf'(x)}{f(x)}\geq\frac{xf'(x/a_i)}{a_if(x/a_i)}$ 
for $x \in \R$, which holds by Assumption~\ref{assumption_mod_price} since
$ |x| \le |x|/a_i$.  
\end{proof}

%% \begin{lemma}
%% \label{lemma_3_price_mod}
%% If $\hat{a}_i<1$, the function
%% $h_{P,i}(x)=x\frac{\left[\frac{1}{a_i}f(\frac{x}{a_i})-\hat{a}_if(x)\right]}{f\left(\frac{x}{a_i}\right)-\hat{a}_if(x)}$
%% is well-defined and increasing.
%% \end{lemma}

%% \begin{proofL}{lemma_3_price_mod}
%% First, we can note that
%% \begin{equation*}
%% h_{P,i}(x)=x\left[1+\frac{a^{-1}_i}{1-\hat{a}_i\frac{f(x)}{f(\frac{x}{a_i})}} \right]  
%% \end{equation*}

%% We see that the denominator of $h_{P,i}$ is positive by the second part of Lemma $\ref{lemma_2_price_mod}$.
%% Using the the third point of Lemma $\ref{lemma_2_price_mod}$ the function $h_{P,i}$ is positive and is nondecreasing on
%% $\mathbb{R}_{+}$ and nonincreasing on $\mathbb{R}_{-}$ as a function of
%% $x$. Therefore, we get the result.
%% \end{proofL}  

\begin{proofT}{thm_price_mod}
We remark that the cost~\eqref{eq_to_min_price_mod_gral} can be written as follows:
\begin{eqnarray*}
C^P ({\bm \xi}, {\bm
    t}) & = & \lambda(t_N)\tilde{F}\left( F^{-1}\left(F(D_N)+\frac{\xi_N}{\lambda(t_N)}\right)\right) \\
& + & \sum^{N-1}_{n=0}\lambda(t_n)\left[ \tilde{F}\left( F^{-1}\left(F(D_n)+\frac{\xi_n}{\lambda(t_n)} \right)\right)-\frac{\lambda(t_{n+1})}{\lambda(t_n)}\tilde{F}\left(a_{n+1}F^{-1}\left(F(D_n)+\frac{\xi_n}{\lambda(t_n)} \right) \right)\right].
\end{eqnarray*}
Since $\tilde{F}$ is convex by Lemma~\ref{lemma_2_price_mod} and $\tilde{F}(0)=0$, we have
$\tilde{F}\left(a_{n+1}x\right)\leq a_{n+1} \tilde{F}(x)$, for $x\in
\R$ and thus
\begin{eqnarray*}
C^P ({\bm \xi}, {\bm
    t})& \geq & \lambda(t_N)\tilde{F}\left( F^{-1}\left(F(D_N)+\frac{\xi_N}{\lambda(t_N)}\right)\right)
 +   \sum^{N-1}_{n=0}\lambda(t_n)\tilde{F}\left(
   F^{-1}\left(F(D_n)+\frac{\xi_n}{\lambda(t_n)}
   \right)\right)(1-\hat{a}_{n+1}). %%\underset{|x|\rightarrow \infty}\longrightarrow\infty \\
\end{eqnarray*}
In particular $C^P ({\bm \xi}, {\bm t})\ge 0$, since $\tilde{F}\ge
0$ and $\hat{a}_{n+1}< 1$ by Assumption~\eqref{assumption_mod_price}. Besides,
by setting $T({\bm
  \xi})=\left(\frac{\xi_0}{\lambda(t_0)},D_1+\frac{\xi_1}{\lambda(t_1)},\dots,D_N+\frac{\xi_N}{\lambda(t_N)}\right)$,
we can easily check that $|T({\bm  \xi})| \underset{|{\bm  \xi} |\rightarrow
  +\infty}\rightarrow +\infty$, which gives immediately that $C^P ({\bm \xi}, {\bm
    t})\underset{|{\bm  \xi} |\rightarrow
  +\infty}\rightarrow +\infty$ since $\tilde{F}(x)\underset{|x |\rightarrow
  +\infty}\rightarrow +\infty$.

Thus, there must be at least one minimizer of~$C^P ({\bm \xi}, {\bm t})$ on
$\{ {\bm
  \xi}\in \R^{N+1}, \sum_{i=0}^N \xi_i=-\vi
\}$, and
we denote by~$\nu$ a Lagrange multiplier such that $\frac{\partial
  C^P}{\partial \xi_i}=\nu$. By Lemma~\ref{lemma_1_price_model_general_case} we obtain:
\begin{equation*}
\nu = h_{P,i+1}(D_{i+1}),\;\; i=0,\ldots ,N-1.  
\end{equation*}
We also have $\frac{\partial
  C^P}{\partial
  \xi_N}=F^{-1}\left(F(D_N)+\frac{x_N}{\lambda(t_N)}\right)=\nu$, and we get ($i=1,\ldots ,N-1$):
\begin{equation*}
\xi^\star_0=\lambda(t_0)F\left(\frac{h^{-1}_{P,1}(\nu)}{a_{1}}\right), \ 
\xi^\star_{i}=\lambda(t_i)\left[
  F\left(\frac{h^{-1}_{P,i+1}(\nu)}{a_{i+1}}\right)-F\left(h^{-1}_{P,i}(\nu)\right)\right], \   
\xi^\star_N=\lambda(t_N)\left[ F(\nu)-F(h^{-1}_{P,N}(\nu))\right].  
\end{equation*}
Besides, we have
\begin{equation}
\label{constraint model 2}  
\lambda(t_N)F(\nu)+\sum^N_{i=1}\lambda(t_{i-1})\left[F\left(\frac{h^{-1}_{P,i}(\nu)}{a_{i}}\right)-\frac{\lambda(t_i)}{\lambda(t_{i-1})}F(h^{-1}_{P,i}(\nu))\right] =-\vi.
\end{equation}
Since $F$ is increasing bijective on~$\mathbb{R}$ and the function
$y\mapsto F(y)-\frac{\lambda(t_i)}{\lambda(t_{i-1})} F(a_i y)$ is increasing
(its derivative is positive by Lemma~\ref{lemma_2_price_mod}), there is a
unique solution to~\eqref{constraint model 2}, and $\nu$ has the same sign as $-\vi$. Thus~${\bm \xi}^\star$ is the
unique optimal strategy. Moreover, the initial and the last trade have the
same sign as~$-\vi$ since $\sgn(\nu)h_{P,N}(\nu)\ge |\nu|$.
\end{proofT}

We now prepare the proof of Theorem~\ref{thm_price_mod_cont}. For sake of
clearness, we will work under assumption~$(i)$ and assume
that $\rho_t\left(1+\frac{xf'(x)}{f(x)}\right)-\eta_t>0$ for any $x\in \R$ and
that $h_{P,t}$ is bijective and increasing. However, a close look at the proof
below is sufficient that the same arguments also work under assumption~$(ii)$.

Contrary to model~$V$, it is more
convenient to work with the process~$D$ rather than~$E$ (both are related by~$D_t=F^{-1}(E_t/\lambda(t))$. We introduce for $0\le t\le T$,
\begin{equation}\label{def_coutP}
  C^P(t,T,D_t,X_t)=\lambda(t)\left[G(\zeta_t)
  - \tilde{F}(D_t) \right]+\int_t^T \zeta_u \xi_u du
  +\lambda(T)[\tilde{F}(\nu)-G(\zeta_T)],
\end{equation}
where
\begin{eqnarray}
 &&\nu\in \R, s.t. -E_t+\int_t^T \lambda(u)\left[ \rho_u h_{P,u}^{-1}(\nu)
   f(h_{P,u}^{-1}(\nu))-\eta_u F\left(h_{P,u}^{-1}(\nu)\right) \right]du
 +\lambda(T)F(\nu)=-X_t, \label{def_nu_costP}\\
&&\zeta_u=h_{P,u}^{-1}(\nu), \ \xi_u=\lambda(u)f(\zeta_u)[\frac{d\zeta_u}{du}+\rho_u\zeta_u].
\end{eqnarray}
Let us observe that $x\mapsto \rho_u x f(x)- \eta_uF(x)$ is increasing:
its derivative is equal
to~$f(x)\left(\rho_u\left(1+\frac{xf'(x)}{f(x)}\right)-\eta_u \right)$ and is
positive by assumption. Therefore, the left hand side of~\eqref{def_nu_costP}
is an increasing bijection on~$\R$ and there is a unique solution~$\nu$
to~\eqref{def_nu_costP}. The function $C^P(t,T,D_t,X_t)$ represents the minimal cost to liquidate~$X_t$
shares on~$[t,T]$ given the current state~$D_t$. We have in
particular that
$C^P(T,T,D_T,X_T)=\lambda(T)\left[G\left(\frac{E_T-X_T}{\lambda(T)} \right)- G\left(
    \frac{E_T}{\lambda(T)} \right)\right], $
which is the cost of selling~$X_T$ shares at
time~$T$. Besides, an integration by parts gives that
\begin{equation}\label{coutP_IPP}
  C^P(t,T,D_t,X_t)=-\lambda(t) \tilde{F}(D_t)
  +\int_t^T  \lambda(u)\left[
    \rho_u f(\zeta_u)\zeta_u^2 -\eta_u \tilde{F}(\zeta_u) \right]du+ \lambda(T)\tilde{F}(\nu).
\end{equation}
The function $\zeta \mapsto \rho_u f(\zeta)\zeta^2 -\eta_u
\tilde{F}(\zeta)$ is nonnegative: it vanishes for $\zeta=0$ and
its derivative is equal to $\zeta f(\zeta)\left(\rho_u\left(2+\frac{\zeta
      f'(\zeta)}{f(\zeta)} \right)   -\eta_u \right)$ and has the same sign as~$\zeta$ by
assumption. Since $\tilde{F}\ge 0$, this gives
\begin{equation}\label{cout_positifP}
  C^P(0,T,0,\vi) \ge 0.
\end{equation}

\begin{proofT}{thm_price_mod_cont}
  Let $(X_t, 0\le t\le T+)$ denote an admissible strategy that
liquidates~$\vi$. We consider $(E_t,0\le t\le T+)$ the solution of 
$dE_t=dX_t+\eta_tE_tdt-\rho_t\lambda(t)f(F^{-1}(E_t/\lambda(t)))F^{-1}(E_t/\lambda(t))dt$,
$D_t=F^{-1}(E_t/\lambda(t))$, $\nu_t$ the solution of~\eqref{def_nu_costP}
and $\zeta_t=h_{P,t}^{-1}(\nu_t)$. We set
$$C_t=\int_0^t  
  D_s dX^c_s +\sum_{0\le s< t}
  \lambda(s)\left[G\left(\frac{E_s+\Delta X_s}{ \lambda(s)}
    \right)-G\left(\frac{E_s}{ \lambda(s)} \right)\right] +C^P(t,T,D_t,X_t). $$
  Let us observe that $C_T=C^P(X)$ and $C_0=C^P(0,T,0,\vi)$. We will
  show that $dC_t\ge 0$, and that $dC_t=0$ holds only for~$X^\star$. This will
  in particular prove that $C^P(X)\ge 0$ from~\eqref{cout_positifP}.

Let us first consider the case of a jump $\Delta X_t>0$. Then, we have
$$\Delta C_t=\lambda(t)\left[G\left(\frac{E_t+\Delta X_t}{ \lambda(t)}
    \right)-G\left(\frac{E_t}{ \lambda(t)} \right)\right]+C^P(t+,T,D_{t+},X_{t+})-C^P(t,T,D_t,X_t).$$
Since $\Delta E_t =\Delta X_t$, we have $\nu_t=\nu_{t+}$
from~\eqref{def_nu_costP} and then $\Delta C_t=0$ since $\tilde{F}(D_t)=G(E_t/\lambda(t))$.
 Now, let us calculate $dC_t$. We set
 \begin{eqnarray*}
    \tilde{C}(t,T,D_t,X_t,v)
   &=&\lambda(T)\tilde{F}(v) -\lambda(t)\tilde{F}\left(D_t\right)
  +\int_t^T  \lambda(u)\left[ \rho_u f(h_{P,u}^{-1}(v))
  h_{P,u}^{-1}(v)^2   - \eta_u \tilde{F}(h_{P,u}^{-1}(v)) \right]du .
\end{eqnarray*}
Since $dD_t^c=-\rho_tD_t dt +\frac{dX^c_t}{\lambda(t)f(D_t)}$, we have from~\eqref{coutP_IPP}:
\begin{eqnarray*}dC_t&=&D_t dX^c_t-
  \lambda'(t) \tilde{F}(D_t) dt +\lambda(t) \rho_t f(D_t)D_t^2dt-D_t dX^c_t 
-\lambda(t)[\rho_t f(\zeta_t) \zeta_t^2 -\eta_t \tilde{F}(\zeta_t)]dt  \\&&+
\frac{\partial
  \tilde{C}}{\partial v}(t,T,D_t,X_t,\nu_t) d\nu_t.
\end{eqnarray*}
Since $d(E_t-X_t)=\lambda(t)\left[\eta_tF(D_t)-\rho_t D_t f(D_t) \right]dt$,
we get from~\eqref{def_nu_costP}
\begin{eqnarray}\label{calcul_dnu_P}
\left[\int_t^T \lambda(u)(h_{P,u}^{-1})'(\nu_t) \left[
    (\rho_u-\eta_u)f(h_{P,u}^{-1}(\nu_t))+ \rho_u h_{P,u}^{-1}(\nu_t)     
   f'(h_{P,u}^{-1}(\nu_t))
    \right]du
 +\lambda(T)f(\nu_t)\right]d\nu_t \hspace{2cm}&&\\-\lambda(t)\left[ \rho_t h_{P,t}^{-1}(\nu_t)
   f(h_{P,t}^{-1}(\nu_t))-\eta_t F\left(h_{P,t}^{-1}(\nu_t)\right) \right]dt
=\lambda(t)\left[\eta_tF(D_t)-\rho_t D_t f(D_t) \right]dt. && \nonumber
\end{eqnarray}
On the other hand, we have
\begin{eqnarray*}
  \partial_v\tilde{C}(t,T,E_t,D_t,v)&=&\lambda(T) v f(v)+\int_t^T
\lambda(u) (h_{P,u}^{-1})'(v) h_{P,u}^{-1}(v) 
\left[ (2\rho_u-\eta_u)f(h_{P,u}^{-1}(v))+ \rho_u h_{P,u}^{-1}(v)     
   f'(h_{P,u}^{-1}(v))\right]du \\
&=&v\left[\lambda(T)  f(v)+\int_t^T\lambda(u) (h_{P,u}^{-1})'(v) 
\left( (\rho_u-\eta_u)f(h_{P,u}^{-1}(v))+ \rho_u h_{P,u}^{-1}(v)     
   f'(h_{P,u}^{-1}(v))\right)
du \right],
\end{eqnarray*}
and we get $\frac{\partial
  \tilde{C}}{\partial v}(t,T,D_t,X_t,\nu_t) d\nu_t
=\lambda(t)\nu_t[\eta_t(F(D_t)-F(\zeta_t))
+\rho_t(\zeta_tf(\zeta_t)-D_tf(D_t))].$ We finally obtain:
\begin{eqnarray}
dC_t&=&\lambda(t) \psi_t(\zeta_t)dt, \text{ with } \\
\psi_t(\zeta)&=&\eta_t(\tilde{F}(\zeta)-\tilde{F}(D_t)) +\rho_t
  (D_t^2f(D_t)-\zeta^2f(\zeta))+h_{P,t}(\zeta) \left(\eta_t(F(D_t)-F(\zeta))
+\rho_t(\zeta f(\zeta)-D_tf(D_t)) \right). \nonumber
\end{eqnarray}
We have $\psi_t(D_t)=0$ and get  that $\psi'_t(\zeta)=h_{P,t}'(\zeta)\left[\eta_t(F(D_t)-F(\zeta))
+\rho_t(\zeta f(\zeta)-D_tf(D_t)) \right]$ by simple calculations. On the one
hand, we have $h_{P,t}'(\zeta)>0$. On the other hand, the bracket is positive
on $\zeta>D_t$ and negative on $\zeta<D_t$ since its derivative is equal to
$(\rho_t-\eta_t)f(\zeta)+\rho_t\zeta f(\zeta)$, which is positive by
assumption. Thus, $D_t$ is the unique minimum of~$\psi_t$: $\psi_t(D_t)=0$ and
$\psi_t(\zeta)>0$ for $\zeta \not = D_t$.

Thus, if $X$ is an optimal strategy, we necessarily have $\zeta_t=D_t$,
$dt$-a.e. From~\eqref{calcul_dnu_P}, we get
$$\left[\int_t^T \lambda(u)(h_{P,u}^{-1})'(\nu_t) \left[
    (\rho_u-\eta_u)f(h_{P,u}^{-1}(\nu_t))+ \rho_u h_{P,u}^{-1}(\nu_t)     
   f'(h_{P,u}^{-1}(\nu_t))
    \right]du
 +\lambda(T)f(\nu_t)\right]d\nu_t=0,$$
and thus $d\nu_t=0$ since $(h_{P,u}^{-1})'$ and $x\mapsto (\rho_u-\eta_u)f(x)+
\rho_u x f'(x)$ are positive functions by assumption. We get that $\nu_t=\nu$,
where $\nu$ is the solution of~\eqref{def_nu_vi_pr}. In particular, we have
$\Delta X_0=\lambda(0)F(D_{0+})=\lambda(0)F(h_{P,0}^{-1}(\nu))=\Delta X^\star_0$ and
then~$X=X^\star$. This gives the uniqueness of the optimal strategy. Last, 
 $\xi^\star_0$ has the same sign as~$-\vi$ since $\nu$ and $-\vi$ have the same sign.
\end{proofT}
\begin{proofC}{cor_price_mod_cont}
By Assumption~\ref{assumption_mod_price} we have $\rho_t-\eta_t>0$, $xf'(x)\ge
0$ and $x\partial_x(\frac{xf'(x)}{f(x)})\le 0$, which gives:
$$h_{P,t}'(x)=\frac{\left(
    \rho_t\left(2+\frac{xf'(x)}{f(x)}\right)-\eta_t \right)\left(
    \rho_t\left(1+\frac{xf'(x)}{f(x)}\right)-\eta_t \right)-\rho_t^2 x\partial_x(\frac{xf'(x)}{f(x)})} { \left(
    \rho_t\left(1+\frac{xf'(x)}{f(x)}\right)-\eta_t \right)^2}>0.$$
Also, we have $\sgn(x) h_{P,t}(x)\ge |x|$, and  $h_{P,t}$ is thus bijective
on~$\R$. We deduce that $\sgn(x) h^{-1}_{P,t}(x)\le |x|$, which gives that the
last trade  $\xi^\star_T$ has the same sign as~$-\vi$.

Let us assume moreover that~\eqref{cond_ttpms_p_bs} holds. Let $\gamma_t= \frac{\lambda(t)f(\zeta_t) \zeta_t}{h'_{P,t}(\zeta_t)\left(
    1+\frac{\zeta_tf'(\zeta_t)}{f(\zeta_t)}-\frac{\eta_t}{\rho_t} \right)^2
}>0$. Then,
\begin{eqnarray*}
  \xi_t&=&\gamma_t \left[\frac{\rho_t'\eta_t-\rho_t\eta_t'}{\rho_t^2}+\rho_t\left(
    1+\frac{\zeta_tf'(\zeta_t)}{f(\zeta_t)}-\frac{\eta_t}{\rho_t} \right)\left(
    2+\frac{\zeta_tf'(\zeta_t)}{f(\zeta_t)}-\frac{\eta_t}{\rho_t}
  \right)-\zeta_t \partial_x(\frac{xf'(x)}{f(x)})|_{x= \zeta_t} \right] \\
&\ge &\gamma_t \left[\frac{\rho_t'\eta_t-\rho_t\eta_t'}{\rho_t^2}+\rho_t\left(
    1-\frac{\eta_t}{\rho_t} \right)\left(
    2-\frac{\eta_t}{\rho_t}
  \right)\right] \text{ by Assumption~\ref{assumption_mod_price}.}\\
&=&\gamma_t \left(\frac{2\rho_t-\eta_t}{\rho_t}\right)^2\left[\left(\frac{\rho_t-\eta_t}{2\rho_t-\eta_t}\right)'+\rho_t\left(\frac{\rho_t-\eta_t}{2\rho_t-\eta_t}\right) \right]\ge 0 \text{ by~\eqref{cond_ttpms_p_bs}.} \qedhere
\end{eqnarray*}
\end{proofC}

%\newpage
\bibliography{ref_papier_lob_time}
\bibliographystyle{plain}
    
\end{document}